\documentclass[aps,prx,reprint,superscriptaddress,nofootinbib,floatfix]{revtex4-2}

\usepackage[T1]{fontenc}
\usepackage{lmodern}
\usepackage{amsmath,amssymb,amsthm,mathtools,bm}
\usepackage{booktabs}
\usepackage{graphicx}
\usepackage{microtype}
\usepackage{xcolor}
\usepackage[colorlinks=true,citecolor=blue,linkcolor=blue,urlcolor=blue]{hyperref}
\hypersetup{
 pdftitle={Uniformly Stable Minimal Weyl--Heisenberg Measurements Approaching the SIC Benchmark},
 pdfauthor={Xiuwu Zhu and Yu Wang}
}

\newtheorem{theorem}{Theorem}
\newtheorem{proposition}[theorem]{Proposition}
\newtheorem{lemma}[theorem]{Lemma}
\newtheorem{corollary}[theorem]{Corollary}
\theoremstyle{definition}

\theoremstyle{remark}

\newcommand{\C}{\mathbb C}
\newcommand{\F}{\mathbb F}
\newcommand{\Z}{\mathbb Z}
\newcommand{\E}{\mathbb E}
\newcommand{\Prob}{\mathbb P}
\newcommand{\Tr}{\operatorname{Tr}}

\newcommand{\Spec}{\operatorname{Spec}}
\newcommand{\ket}[1]{\lvert #1\rangle}
\newcommand{\bra}[1]{\langle #1\rvert}

\newcommand{\ketbra}[2]{\lvert #1\rangle\!\langle #2\rvert}
\newcommand{\abs}[1]{\left\lvert #1\right\rvert}
\newcommand{\norm}[1]{\left\lVert #1\right\rVert}
\newcommand{\HS}{\mathrm{HS}}
\newcommand{\cM}{\mathcal M}
\newcommand{\cS}{\mathcal S}
\newcommand{\cH}{\mathcal H}

\begin{document}

\title{Uniformly Stable Minimal Weyl--Heisenberg Measurements Approaching the SIC Benchmark}

\author{Xiuwu Zhu}
\affiliation{Beijing Key Laboratory of Topological Statistics and Applications for Complex Systems, Beijing Institute of Mathematical Sciences and Applications, Beijing 101408, China}

\author{Yu Wang}
\email{ming-jing-happy@163.com}
\affiliation{Hetao Institute of Mathematics and Interdisciplinary Sciences, Shenzhen, Guangdong 518017, China}
\affiliation{Beijing Key Laboratory of Topological Statistics and Applications for Complex Systems, Beijing Institute of Mathematical Sciences and Applications, Beijing 101408, China}

\date{August 12, 2026}

\begin{abstract}
Informational completeness (IC) guarantees that an inverse exists, not that it is statistically well conditioned.  For minimal rank-one Weyl--Heisenberg (WH) measurements, covariance makes the nonidentity projector-Gram spectrum proportional to the fiducial's ambiguity intensities, with eigenvalues \(d|\chi_\phi(u)|^2\), turning stability into an explicit worst-direction design problem; write \(\lambda\) for its smallest nonidentity eigenvalue.  Haar fiducials are IC almost surely while \(\mathbb E[\lambda^{-1}]=\infty\), and an explicit geometric family used to establish balanced informationally complete measurements in every dimension has a normalized spectral floor bounded above by an exponentially decaying envelope.  We then construct a hierarchy of minimal measurements.  A cyclic family with exactly \(d^2\) outcomes in every integer dimension has floors \(\Theta(d^{-3})\) for odd \(d\) and \(\Theta(d^{-5})\) for even \(d\); a finite-field family for \(q=2^m\) obeys the uniform bound \(\lambda\ge4/9\).  Our main result treats every prime-power dimension of characteristic \(p\ge5\).  A balanced one-coordinate perturbation repairs the zero ambiguity axis of a cubic Alltop state, gives an attained floor uniformly bounded below by a positive constant, and confines the entire nonidentity spectrum to \([L_q,U_q]\) with \(U_q/L_q\to1\).  Its SIC-normalized minimum tends to one, and \(\lambda(\phi_q)/\Lambda_q^\star\to1\) for the global finite-field WH max--min optimum \(\Lambda_q^\star\), without assuming SIC existence.  The complete spectrum determines the exact finite-sample Hilbert--Schmidt error of canonical linear inversion at \(I/d\), while its lower edge controls local Fisher efficiency and canonical-shadow bounds.
\end{abstract}

\maketitle

\section{Introduction}
\label{sec:intro}
An informationally complete (IC) measurement uniquely determines a quantum state from its outcome statistics.  But uniqueness does not imply stability.  An IC measurement may resemble an invertible but nearly singular matrix: statistical errors or small model perturbations can be strongly amplified even though the inverse exists.  The central question of this work is quantitative: for minimal Weyl--Heisenberg measurements, how stable can explicit constructions be, and how closely can they approach the SIC max--min benchmark without assuming that an exact SIC exists?

Weyl--Heisenberg (WH) covariance makes this question explicitly computable.  The cyclic phase space \(\Z_d^2\) has labels \(u=(a,b)\), with \(D_u=X^aZ^b\) applying a shift and a phase modulation; overall displacement phases are immaterial.  For a normalized fiducial \(\ket\phi\), set
\begin{equation}
 \Pi_u=D_u\ketbra\phi\phi D_u^\dagger,
 \qquad
 G_\phi^\Pi[u,v]=\Tr(\Pi_u\Pi_v).
 \label{eq:intro-projector-Gram}
\end{equation}
Throughout the paper, a bare \(G_\phi\) means this projector Gram matrix unless an effect Gram matrix \(G_\phi^E\) is displayed explicitly.  Its smallest nonidentity eigenvalue is
\begin{equation}
 \lambda(\phi):=\lambda_{\min}\!\left(
 G_\phi^\Pi\big|_{\boldsymbol{1}^{\perp}}
 \right).
 \label{eq:intro-lambda}
\end{equation}
Here \(\boldsymbol{1}=(1,\ldots,1)^T\in\C^{d^2}\) is the constant phase-space vector, and \(\boldsymbol{1}^\perp\) corresponds to the traceless operator sector.  Phase-space characters diagonalize \(G_\phi^\Pi\), with eigenvalues \(d|\chi_\phi(u)|^2\) up to a symplectic relabeling \cite{Goldberger2022}.  Zeros therefore mark loss of IC, small coefficients mark weakly resolved operator directions, and a flat nonidentity spectrum is the SIC endpoint.  This criterion targets a different quantity from the pairwise-coherence upper bounds commonly used in approximate-SIC constructions \cite{Klappenecker2005,CaoDeng2024}: an upper bound on coherence alone need not lower-bound the weakest ambiguity coefficient, whereas our main construction controls the entire nonidentity spectrum.

We use the SIC-normalized score
\begin{equation}
 \eta(\phi)=\frac{d+1}{d}\lambda(\phi)
 =(d+1)\min_{u\ne0}\abs{\bra\phi D_u\ket\phi}^2.
 \label{eq:intro-eta}
\end{equation}
For a single orbit, \(\eta>0\) is equivalent to IC.  A dimension-indexed family is uniformly spectrally stable when \(\eta\) has a positive dimension-independent lower bound.  Figure~\ref{fig:story}(a) summarizes the hierarchy
\[
 \mathrm{SIC}\ \Longrightarrow\
 \text{uniform spectral stability}\ \Longrightarrow\
 \mathrm{IC},
\]
whose converses fail.  BIC is a structural property and does not lie on this stability hierarchy.

Operator-frame theory identifies the lower frame bound as the quantity controlling canonical inversion and reconstruction-error bounds.  The same edge governs canonical-estimator second moments in IC-POVM extensions of classical shadows \cite{Huang2020,Acharya2021,Innocenti2023} and, at the maximally mixed state, the least classical-to-quantum Fisher-information ratio \cite{SainiFisher2026}.  Completeness stability is maximized without structural restrictions by weighted complex projective \(2\)-designs \cite{Saini2026}; within a minimal equally weighted rank-one family, the symmetric endpoint is a SIC-POVM \cite{Renes2004,Scott2006}.  Exact SICs are known in many dimensions, but unconditional existence in every dimension remains open \cite{Horodecki2022,Bengtsson2025,ApplebyFlammiaKopp2025}.  Our max--min formulation therefore asks how closely explicit WH measurements can reach that endpoint without relying on SIC existence.

Balanced informationally complete (BIC) measurements are tight, minimal, equal-weight rank-one IC measurements and attain maximal ideal device-independent randomness in every dimension \cite{Farkas2026}.  Every complete full rank-one WH orbit is automatically BIC, yet BIC imposes no quantitative lower bound on its weakest operator-space direction.  Indeed, the explicit geometric WH family used to prove all-dimensional BIC existence has a SIC-normalized spectral floor bounded above by \(O(d\,2^{-d})\).  Haar fiducials provide a complementary separation: they are IC almost surely, and every fixed nonidentity Gram eigenvalue has the SIC mean, but \(\mathbb E[\lambda^{-1}]=\infty\).  Generic IC is not enough.

We next build stability in stages.  In every integer dimension, an explicit cyclic fiducial gives a minimal equal-weight rank-one IC POVM with exactly \(d^2\) outcomes and projector-Gram floors \(\Theta(d^{-3})\) for odd \(d\) and \(\Theta(d^{-5})\) for even \(d\).  A separate finite-field fiducial for every \(q=2^m\) raises the floor uniformly to at least \(4/9\), covering all multi-qubit Hilbert-space dimensions, although its spectrum is not asymptotically flat.

The strongest result begins with a cubic Alltop state over \(\F_q\).  Its ambiguity profile is flat except for one zero axis.  A single-coordinate perturbation repairs this zero set, and balancing the repaired-axis amplitude against the bulk distortion selects \(t_q\asymp q^{-1/4}\).  For every prime power of characteristic \(p\ge5\), the resulting nonidentity spectrum lies in \([L_q,U_q]\), where \(L_q\) is uniformly positive and \(U_q/L_q\to1\).  Its SIC-normalized minimum tends to one.  More strongly, if
\[
 \Lambda_q^\star=\max_{\norm\phi=1}\lambda(\phi)
\]
is the global finite-field WH max--min optimum, then \(\lambda(\phi_q)/\Lambda_q^\star\to1\).  This asymptotic optimality is unconditional on SIC existence and does not identify a finite-\(q\) global optimizer.

The projector-Gram spectrum also supplies the operational interpretation used below.  Its lower edge controls worst-direction inverse amplification, local Fisher efficiency at \(I/d\), and canonical-shadow second-moment bounds; the full spectrum determines the exact finite-sample Hilbert--Schmidt MSE of canonical linear inversion there.  The complete WH orbits can also instantiate the ideal BIC-based randomness protocol, although the spectral floor alone does not order its nonideal robustness.
\section{Spectral stability of minimal WH measurements}
\label{sec:framework}

\subsection{Exact spectral interface}

Let \(d\ge2\), \(\omega=e^{2\pi i/d}\), and
\begin{equation}
 X\ket x=\ket{x+1\bmod d},
 \qquad
 Z\ket x=\omega^x\ket x.
 \label{eq:cyclic-XZ}
\end{equation}
For \(u=(a,b)\in\Z_d^2\), set \(D_u=X^aZ^b\).  Overall displacement phases will never matter.  A normalized fiducial \(\ket\phi\in\C^d\) generates
\begin{equation}
 \Pi_u=D_u\ketbra\phi\phi D_u^\dagger,
 \qquad
 E_u=\frac1d\Pi_u.
 \label{eq:orbit}
\end{equation}
WH twirling gives \(\sum_u\Pi_u=dI\), so \(\{E_u\}\) is a rank-one POVM.  Define
\begin{equation}
 \chi_\phi(a,b)=\bra\phi X^aZ^b\ket\phi
 \label{eq:chi}
\end{equation}
and the ambiguity intensity
\begin{equation}
 g_\phi(a,b)=\abs{\chi_\phi(a,b)}^2.
 \label{eq:ambiguity-intensity}
\end{equation}

The \emph{projector} and \emph{effect} Gram matrices are, respectively,
\begin{align}
 G_\phi^\Pi[u,v]
 &=\Tr(\Pi_u\Pi_v)=g_\phi(v-u),
 \label{eq:projector-Gram}\\
 G_\phi^E[u,v]
 &=\Tr(E_uE_v)=\frac1{d^2}G_\phi^\Pi[u,v].
 \label{eq:effect-Gram}
\end{align}
Here \(u,v\in\Z_d^2\) label the \(d^2\) elements of the WH orbit, so
\(G_\phi^\Pi\) is a \(d^2\times d^2\) matrix whose rows and columns are
indexed by discrete phase-space points.  Since each entry depends only on
the phase-space difference \(v-u\), \(G_\phi^\Pi\) acts as a convolution
operator on \(\Z_d^2\).

As declared in the Introduction, bare \(G_\phi\) denotes \(G_\phi^\Pi\), and
\begin{equation}
 \lambda(\phi)
 :=\lambda_{\min}\!\left(
 G_\phi^\Pi\big|_{\boldsymbol{1}^{\perp}}
 \right)
 \label{eq:lambda-definition}
\end{equation}
is its smallest nonidentity eigenvalue, including zero when the orbit is incomplete.  Here \(\boldsymbol{1}\) is the constant phase-space character, and \(\boldsymbol{1}^{\perp}\) corresponds to the traceless operator sector.  The matrix \(G_\phi^\Pi\) is the Hilbert--Schmidt Gram matrix of the \(d^2\) projectors, rather than the ordinary state-vector Gram matrix of the \(d^2\) orbit states.  It records how the measurement operators overlap and therefore how evenly operator space is resolved.

For \(d^2\) projectors, informational completeness is equivalent to nonsingularity of their Hilbert--Schmidt Gram matrix.  Earlier covariant and dynamical-tomography work used related nonzero-characteristic-function and nonzero-determinant criteria \cite{DAriano2004,CaoDengWang2024}.  Finite Gabor operator-frame bounds identify the same lower and upper ambiguity scales \cite{BojarovskaFlinth2016}, and Goldberger \emph{et al.} explicitly diagonalized the rank-one-projector Gramian by the two-dimensional Fourier transform \cite{Goldberger2022}.  We restate that established eigensystem in our displacement convention and projector-Gram normalization.

WH covariance makes the projector Gram matrix depend only on phase-space differences, so it acts as a convolution operator.  Fourier characters are therefore its natural eigenmodes; equivalently, in the phase-space coefficient basis, the diagonalizing transform is the two-dimensional discrete Fourier transform \(F_d\otimes F_d\).

\begin{theorem}[WH projector-Gram spectrum]
\label{thm:WH-spectrum}
For $m,n\in\Z_d$, define the two-dimensional Fourier vector
\begin{equation}
 \ket{f_{m,n}}
 =
 \frac{1}{d}
 \sum_{a,b\in\Z_d}
 \omega^{ma+nb}\ket{a,b},
 \label{eq:phase-character}
\end{equation}
where $\{\ket{a,b}\}$ denotes the standard basis of the
phase-space coefficient space $\C^{d^2}$.
Equivalently,
\[
 \ket{f_{m,n}}
 =
 \ket{\widetilde m}\otimes\ket{\widetilde n},
 \qquad
 \ket{\widetilde m}
 =
 \frac{1}{\sqrt d}\sum_{a\in\Z_d}\omega^{ma}\ket a.
\]
Then $\ket{f_{m,n}}$ is an eigenvector of $G_\phi$ with
\begin{align}
 \lambda_{m,n}
 &=\sum_{a,b\in\Z_d}
 |\chi_\phi(a,b)|^2\omega^{ma+nb}
 \label{eq:spectrum-DFT}\\
 &=d|\chi_\phi(-n,m)|^2.
 \label{eq:spectrum-explicit}
\end{align}
The $d^2$ vectors $\{\ket{f_{m,n}}\}_{m,n\in\Z_d}$ form the
two-dimensional Fourier eigenbasis, equivalently the columns of
$F_d\otimes F_d$. Consequently,
\begin{equation}
 \Spec(G_\phi)
 =
 \{d|\chi_\phi(u)|^2:u\in\Z_d^2\},
 \label{eq:spectrum-multiset}
\end{equation}
including multiplicities.
\end{theorem}

\begin{proof}[Fourier/circulant proof]
For \(u=(u_1,u_2)\), substitute \(w=v-u\) in the convolution to obtain
\begin{align}
 (G_\phi f_{m,n})(u)
 &=\sum_v g_\phi(v-u)f_{m,n}(v)\notag\\
 &=f_{m,n}(u)\sum_{a,b}g_\phi(a,b)\omega^{ma+nb}.
 \label{eq:convolution-diagonalization}
\end{align}
Character orthogonality makes the \(f_{m,n}\) a complete orthonormal basis, proving Eq.~\eqref{eq:spectrum-DFT}.  To evaluate this Fourier transform, WH orthogonality and
\(
D_{a,b}D_{p,q}D_{a,b}^\dagger
=\omega^{bp-aq}D_{p,q}
\)
give the ambiguity identity
\begin{equation}
 g_\phi(a,b)=\frac1d\sum_{p,q}
 |\chi_\phi(p,q)|^2\omega^{bp-aq}.
 \label{eq:ambiguity-self-Fourier}
\end{equation}
Indeed, expand
\begin{align}
 \Pi_0
 &=\frac1d\sum_{p,q}\overline{\chi_\phi(p,q)}D_{p,q},
 \notag\\
 g_\phi(a,b)
 &=\Tr(\Pi_0D_{a,b}\Pi_0D_{a,b}^\dagger).
 \label{eq:projector-Weyl-expansion}
\end{align}
Fourier summation over \(a,b\) in Eq.~\eqref{eq:ambiguity-self-Fourier} selects \((p,q)=(-n,m)\), yielding Eq.~\eqref{eq:spectrum-explicit} with the stated signs and labels.
\end{proof}

The theorem gives a direct dictionary: \(\chi_\phi(u)=0\) marks a missing operator direction and loss of IC; small \(|\chi_\phi(u)|\) marks a weakly resolved direction and an unstable inverse; and a flat nonidentity spectrum means isotropic resolution at the SIC endpoint.

Appendix~\ref{app:spectrum} gives a complementary synthesis-operator proof, including the zero-eigenvalue case.  The Fourier route explains why phase-space characters are the eigendirections; the synthesis route explains why their Fourier eigenvalues collapse back to the ambiguity profile.  The immediate consequences are
\begin{align}
 \{E_u\}\text{ is IC}
 &\Longleftrightarrow \chi_\phi(u)\ne0
 \quad\text{for every }u\ne0,
 \label{eq:IC-iff}\\
 \lambda(\phi)
 &=d\min_{u\ne0}|\chi_\phi(u)|^2.
 \label{eq:lambda-min}
\end{align}
The identity direction has eigenvalue \(d\).  Weyl orthogonality gives the finite Moyal identity
\begin{equation}
 \sum_{u\in\Z_d^2}|\chi_\phi(u)|^2=d.
 \label{eq:Moyal}
\end{equation}
Since the Gram eigenvalues are \(d|\chi_\phi(u)|^2\), their total sum is \(d^2\).  The identity contribution is \(d|\chi_\phi(0)|^2=d\), so the remaining \(d^2-1\) nonidentity eigenvalues have the fixed sum
\begin{equation}
 \sum_{u\ne0} d|\chi_\phi(u)|^2
 =d^2-d.
 \label{eq:nonidentity-sum}
\end{equation}

\begin{figure*}[t]
 \centering
 \includegraphics[width=0.98\textwidth]{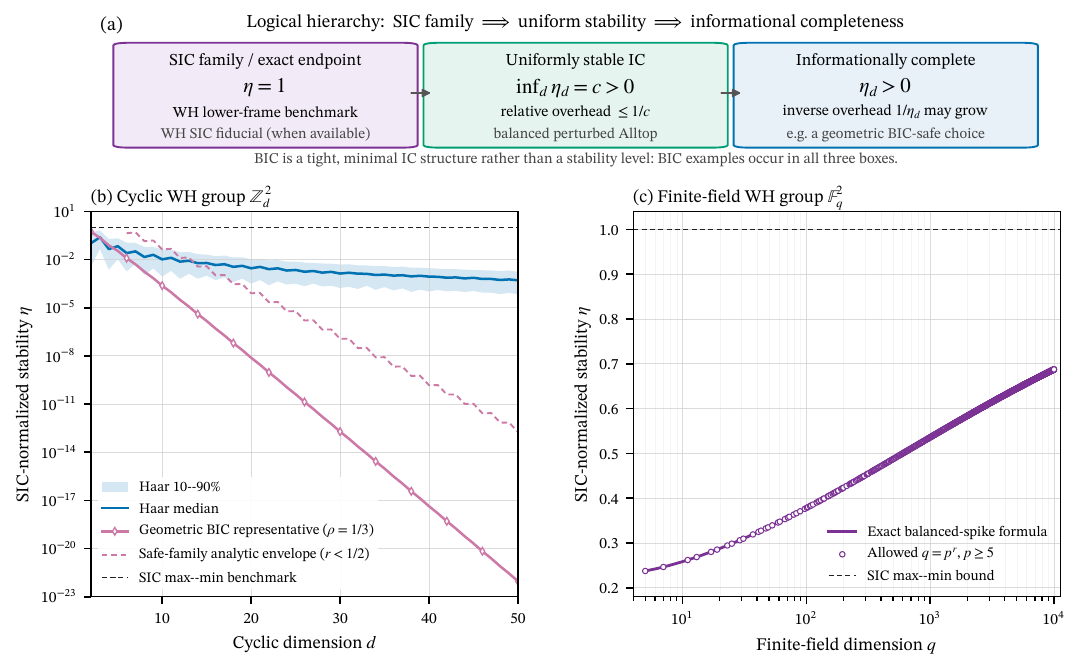}
 \caption{Spectral stability beyond binary informational completeness; larger \(\eta\) is better.  (a) A SIC is the exact \(\eta=1\) endpoint, uniform stability is a dimension-independent relative guarantee, and ordinary IC says only that the inverse exists, so conditioning may degrade.  BIC is a tight minimal-IC structure, not a stability level.  (b) Cyclic \(\Z_d^2\) benchmarks: blue shows the median and 10--90\% range of 2000 Haar fiducials per dimension; pink shows the geometric representative \(\alpha=(1/3)e^{2\pi i t_d}\) from the explicit sufficient IC region, with \(t_d=0\) for odd \(d\) and \(t_d=1/(4d)\) for even \(d\).  Its analytic upper envelope decays exponentially but does not bound all BICs.  (c) Exact balanced-Alltop stability for \(q=p^r\), \(p\ge5\).  This panel uses finite-field phase space \(\F_q^2\), distinct from cyclic \(\Z_q^2\) when \(r>1\).  The finite-\(q\) curve illustrates the proved convergence \(\eta\to1\), slow at the \(q^{-1/4}\) scale; the dashed line is the universal WH-SIC upper endpoint.}
 \label{fig:story}
\end{figure*}

\subsection{Stability scale and the SIC endpoint}

How large can the weakest resolved direction \(\lambda(\phi)\) be?  The fixed Moyal sum turns this stability question into a max--min problem.

\begin{proposition}[Pointwise max--min endpoint]
\label{prop:SIC-endpoint}
Every normalized WH fiducial satisfies
\begin{equation}
 \lambda(\phi)\le\frac d{d+1}.
 \label{eq:SIC-upper}
\end{equation}
Equality holds if and only if
\begin{equation}
 |\chi_\phi(u)|^2=\frac1{d+1}
 \quad\text{for all }u\ne0,
 \label{eq:SIC-condition}
\end{equation}
that is, if and only if its WH orbit is a SIC.
\end{proposition}

\begin{proof}
The total nonidentity spectral weight is fixed, so the largest possible minimum is obtained only when that weight is distributed uniformly across all nonidentity directions.  Quantitatively, the minimum cannot exceed the fixed average \(d/(d+1)\), and equality forces every nonidentity eigenvalue to equal that average.  Equation~\eqref{eq:spectrum-multiset} then gives Eq.~\eqref{eq:SIC-condition}; the converse is immediate.
\end{proof}

A WH SIC is therefore not inserted as an external symmetry target; it emerges as the max--min endpoint of stable minimal WH tomography.  This direct specialization is consistent with the general completeness-stability optimum of Ref.~\cite{Saini2026}.  Proposition~\ref{prop:SIC-endpoint} shows that only a SIC can saturate the bound; the balanced-Alltop family below supplies a non-SIC sequence that approaches it.  We henceforth use the normalized quantity in Eq.~\eqref{eq:intro-eta}, which lies in \([0,1]\).  For a family \(\{\phi_d:d\in\mathcal D\}\), we call
\begin{equation}
 \inf_{d\in\mathcal D}\eta(\phi_d)>0
 \label{eq:uniform-stability}
\end{equation}
\emph{uniform spectral stability}.  This is a property of a dimension-indexed family; ordinary IC is only the pointwise condition \(\eta(\phi_d)>0\).

Three normalizations occur in the literature.  Table~\ref{tab:normalization} records the same spectral stability in three normalizations, not three independent metrics.  Within this minimal rank-one class, a dimension-independent floor for \(G^\Pi\) corresponds to the SIC-scale \(O(d)\) inverse scaling of the physical frame channel, not to a constant channel gap.

\begin{table}[t]
 \caption{The same smallest nonidentity spectral value in three normalizations.  Here \(\lambda=\lambda(\phi)\) always refers to the projector Gram matrix \(G_\phi^\Pi\).}
\label{tab:normalization}
\begin{ruledtabular}
\begin{tabular}{lcc}
Object & Minimum & SIC value\\
\hline
Projector Gram \(G^\Pi\) & \(\lambda\) & \(d/(d+1)\)\\
Effect Gram \(G^E\) & \(\lambda/d^2\) & \(1/[d(d+1)]\)\\
Scaled frame \(\cM\) & \(\lambda/d\) & \(1/(d+1)\)
\end{tabular}
\end{ruledtabular}
\end{table}

Indeed, with
\begin{equation}
 \cS_\phi(A)=\sum_u\Tr(\Pi_uA)\Pi_u,
 \qquad
 \cM_\phi=\frac1d\cS_\phi,
 \label{eq:frame-channel}
\end{equation}
\(\cS_\phi\) is the unscaled frame superoperator and \(\cM_\phi\) is the scaled measurement frame channel.  We have \(\cM_\phi(D_u)=|\chi_\phi(u)|^2D_u\).  For the effects \(E_u=\Pi_u/d\), \(\cM_\phi\) is exactly the scaled frame \(F_s\) of Ref.~\cite{Saini2026}.  Denoting its smallest traceless-sector eigenvalue by \(s(\{E_u\})\),
\begin{equation}
 s(\{E_u\})=\frac{\lambda(\phi)}d,
 \qquad
 \eta(\phi)=(d+1)s(\{E_u\}).
 \label{eq:Saini-relation}
\end{equation}

\subsection{Canonical inversion, Fisher information, and shadows}
The spectrum has complementary operational roles: its smallest nonidentity
eigenvalue controls worst-direction amplification, while the full spectrum
determines the exact mean-squared error (MSE) of canonical linear inversion
at the maximally mixed state, measured in the Hilbert--Schmidt norm.
The same lower edge also fixes the least Fisher-information direction there
and universal second-moment bounds for canonical classical-shadow estimators.

When the orbit is IC, observing outcome \(u\) yields the canonical snapshot
\begin{equation}
 \widehat\rho_u=\cM_\phi^{-1}(\Pi_u),
 \label{eq:snapshot}
\end{equation}
which is unbiased because
\[
 \sum_u \Tr(E_u\rho)\,\Pi_u=\cM_\phi(\rho),
 \qquad
 \E_\rho[\widehat\rho_u]=\rho.
\]
Thus canonical inversion reconstructs the state by applying
\(\cM_\phi^{-1}\) to each observed outcome.

Although \(\cM_\phi\) acts on \(d\times d\) matrices, it is a linear
operator on the \(d^2\)-dimensional Hilbert--Schmidt operator space.
To relate its spectrum to the projector Gram spectrum, define the synthesis map
\[
 T:\C^{d^2}\to\mathcal L(\C^d),
 \qquad
 T\ket{u}=\lvert\Pi_u),
\]
where \(\ket{u}\) denotes the standard basis of the \(d^2\)-dimensional
coefficient space.  Then
\begin{equation}
 G_\phi^\Pi=T^\dagger T,
 \qquad
 d\,\cM_\phi=TT^\dagger.
 \label{eq:Gram-frame-relation}
\end{equation}
Indeed,
\[
 (T^\dagger T)_{u,v}
 =(\Pi_u\vert\Pi_v)
 =\Tr(\Pi_u\Pi_v),
 \quad
 TT^\dagger
 =\sum_u\lvert\Pi_u)(\Pi_u\rvert .
\]
The operators \(T^\dagger T\) and \(TT^\dagger\) have the same nonzero
eigenvalues.  Hence, if \(\lambda_v\) is a nonidentity eigenvalue of
\(G_\phi^\Pi\), the corresponding eigenvalue of \(\cM_\phi\) is
\(\lambda_v/d\).  After inversion this eigenvalue becomes \(d/\lambda_v\).
Therefore the weakest measurement direction is the most strongly amplified
direction:
\begin{equation}
 \norm{\cM_\phi^{-1}}_{2\to2}
 =
 \max_{v\ne0}\frac{d}{\lambda_v}
 =
 \frac{d}{\lambda(\phi)}.
 \label{eq:inverse-norm}
\end{equation}
Here
\begin{equation}
 \norm{\mathcal T}_{2\to2}
 =
 \sup_{\substack{A\ne0\\ \Tr A=0}}
 \frac{\norm{\mathcal T(A)}_{\HS}}
      {\norm{A}_{\HS}}
 \label{eq:induced-HS-norm}
\end{equation}
is the largest factor by which \(\mathcal T\) can amplify the
Hilbert--Schmidt size of a traceless operator.  Hence a small
\(\lambda(\phi)\) means that canonical inversion strongly amplifies errors
along at least one operator-space direction.

\begin{corollary}[Exact finite-sample canonical tomography error]
\label{cor:linear-MSE}
Let \(\overline\rho_N=N^{-1}\sum_{j=1}^N\widehat\rho_{u_j}\) be the mean of \(N\) independent canonical snapshots generated from \(\rho_*=I/d\).  If \(\{\lambda_v:v\ne0\}\) are the \(d^2-1\) nonidentity eigenvalues of \(G_\phi^\Pi\), then
\begin{equation}
 \E_{I/d}\norm{\overline\rho_N-I/d}_{\HS}^2
 =\frac1N\sum_{v\ne0}\frac1{\lambda_v}.
 \label{eq:linear-MSE}
\end{equation}
For a WH SIC this becomes
\begin{equation}
 \operatorname{MSE}_{\mathrm{SIC}}
 =\frac{(d^2-1)(d+1)}{Nd}.
 \label{eq:SIC-MSE}
\end{equation}
Consequently, the exact canonical MSE of the WH orbit, relative to the
WH-SIC benchmark, is
\begin{equation}
 R(\phi)
 =
 \frac{\sum_{v\ne0}\lambda_v^{-1}}
 {(d^2-1)(d+1)/d}.
 \label{eq:MSE-ratio}
\end{equation}
\end{corollary}

\begin{proof}
At \(I/d\), all \(d^2\) outcomes have probability \(1/d^2\).  Using \(\sum_u|\Pi_u)(\Pi_u|=d\cM_\phi\), self-adjointness, and the identity-sector eigenvalue one of \(\cM_\phi\),
\begin{equation}
 \frac1{d^2}\sum_u\|\cM_\phi^{-1}(\Pi_u)\|_{\HS}^2
 =\frac1d\Tr_{\HS}(\cM_\phi^{-1})
 =\frac1d+\sum_{v\ne0}\frac1{\lambda_v}.
 \label{eq:MSE-trace-proof}
\end{equation}
Since the estimator is unbiased,
\[
 \E_{I/d}\norm{\widehat\rho_u-I/d}_{\HS}^2
 =
 \E_{I/d}\norm{\widehat\rho_u}_{\HS}^2
 -\norm{I/d}_{\HS}^2,
\]
and
\[
 \norm{I/d}_{\HS}^2=\frac1d.
\]
Therefore the identity-sector contribution \(1/d\) in
Eq.~\eqref{eq:MSE-trace-proof} cancels.  Finally, independence of the
\(N\) centered snapshots gives
\[
 \E_{I/d}\norm{\overline\rho_N-I/d}_{\HS}^2
 =
 \frac1N\sum_{v\ne0}\frac1{\lambda_v}.
\]
The SIC value follows from
\(\lambda_v=d/(d+1)\) for every \(v\ne0\).
\end{proof} 

This equality is scoped to the maximally mixed input, independent samples, canonical linear inversion, and Hilbert--Schmidt MSE; it is not a claim of arbitrary-state or estimator-optimal tomography.

Unlike \(\eta\), which probes the weakest resolved eigendirection, the canonical linear-inversion MSE depends on the full nonidentity Gram spectrum.  More generally, if the nonidentity projector-Gram spectrum lies in \([L,U]\), then
\begin{equation}
 \frac{d}{(d+1)U}\le R(\phi)
 \le\frac{d}{(d+1)L}.
 \label{eq:MSE-spectral-bounds}
\end{equation}
The spectral interval proved below for balanced Alltop makes both bounds tend to one.

\clearpage
\begin{figure}[t]
 \centering
 \includegraphics[width=\columnwidth]{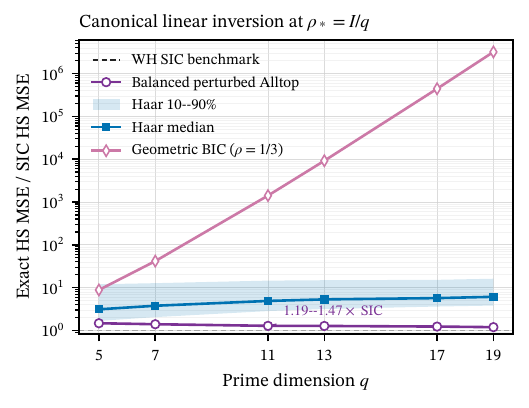}
 \caption{Exact Hilbert--Schmidt MSE ratio \(R\) for canonical linear inversion at \(I/q\) in prime dimensions.  The WH-SIC benchmark is \(R=1\), and larger \(R\) means larger error.  Purple circles show balanced Alltop, blue the median and 10--90\% range over 4000 Haar fiducials, and pink diamonds the geometric representative \(\alpha=1/3\) from the explicit sufficient IC region (all plotted dimensions are odd).  Ratios are evaluated analytically from the nonidentity projector-Gram spectrum and are independent of \(N\); no measurement-shot Monte Carlo is used.  For prime \(q\), cyclic and finite-field WH phase spaces coincide.}
 \label{fig:operational-MSE}
\end{figure}

Figure~\ref{fig:operational-MSE} uses the WH-SIC value \(R=1\) as its benchmark.  Across the displayed primes, the balanced-Alltop ratio decreases from \(1.474\) to \(1.194\), whereas the Haar median rises from \(3.117\) to \(6.108\) and the geometric representative grows from \(8.828\) to \(3.20\times10^6\).  Thus informational completeness alone does not control this canonical finite-sample linear-inversion MSE, even for the fixed maximally mixed input.

Fisher information quantifies how sensitively outcome probabilities respond to an infinitesimal change of the state in a specified parameter direction.  Consider a one-parameter local model around \(\rho_*=I/d\), with traceless Hermitian symmetric logarithmic derivative \(L\); here \(L\) specifies the local operator-space direction in which the state is varied.  The classical Fisher information \(I_C\) is obtained from this particular POVM, whereas \(I_Q\) is the quantum Fisher-information benchmark for the same local model.  The frame formulation of Ref.~\cite{SainiFisher2026} gives the Rayleigh quotient
\begin{equation}
 \frac{I_C}{I_Q}
 =\frac{\langle L,\cM_\phi(L)\rangle_{\HS}}
 {\langle L,L\rangle_{\HS}}.
\end{equation}
Thus \(I_C/I_Q\) is the fraction of locally available information retained by this measurement in direction \(L\).  Minimizing the Rayleigh quotient over nonzero Hermitian traceless \(L\) identifies the weakest resolved operator-space direction and gives
\begin{align}
 \min_{\substack{0\ne L=L^\dagger\\\Tr L=0}}
 \frac{I_C}{I_Q}
 &=\frac{\lambda(\phi)}d,
 \notag\\[-1mm]
 &=\frac{\eta(\phi)}{d+1}.
 \label{eq:Fisher-min}
\end{align}
Indeed, \(\cM_\phi\) preserves Hermiticity, and the Weyl eigenspaces at \(u\) and \(-u\) share the same eigenvalue; a nonzero Hermitian or anti-Hermitian component therefore attains the spectral minimum.
For a SIC this ratio is \(1/(d+1)\) in every traceless direction.  Thus \(\eta\) is exactly the worst-direction Fisher information relative to the WH SIC benchmark at \(\rho_*\), and \(\eta^{-1}\) is the relative worst-direction scalar Cram\'er--Rao-bound overhead there.  This equality is distinct from the state-uniform shadow bounds below.

Because \(\cM_\phi(I)=I\) and \(\cM_\phi\) is self-adjoint, it is trace preserving; so is its inverse.  Hence every canonical snapshot in Eq.~\eqref{eq:snapshot} has unit trace.

The classical-shadow paradigm \cite{Huang2020} extends to IC POVMs and general measurement frames \cite{Acharya2021,Innocenti2023}.  For the canonical IC-shadow estimator used here, the same spectral floor controls second moments and hence variance and sample-complexity bounds.

\begin{corollary}[Canonical-shadow spectral bounds]
\label{cor:shadow}
For a Hermitian observable \(O\), write \(O_0=O-(\Tr O)I/d\) and use the full single-shot estimator
\begin{equation}
 \widehat O_u=\frac{\Tr O}{d}
 +\Tr(O_0\widehat\rho_u).
 \label{eq:full-estimator}
\end{equation}
Then the random traceless contribution \(\widehat o_u=\Tr(O_0\widehat\rho_u)\) satisfies
\begin{align}
 \sup_\rho\E_\rho[\widehat o_u^2]
 &\le\frac d{\lambda(\phi)}\Tr(O_0^2),
 \label{eq:shadow-worst}\\
 \E_{I/d}[\widehat o_u^2]
 &=\frac1d\langle O_0,\cM_\phi^{-1}(O_0)\rangle_{\HS}
 \le\frac1{\lambda(\phi)}\Tr(O_0^2).
 \label{eq:shadow-mixed}
\end{align}
If the nonidentity projector-Gram spectrum lies in \([L,U]\), the exact second moment in Eq.~\eqref{eq:shadow-mixed} lies between \(\Tr(O_0^2)/U\) and \(\Tr(O_0^2)/L\).
\end{corollary}

The proof is in Appendix~\ref{app:shadow}.  The inverse channel is self-adjoint and Hermiticity-preserving, but need not be a positive map.  Equations~\eqref{eq:shadow-worst} and \eqref{eq:shadow-mixed} are uniform spectral bounds, not exact observable-by-observable variance formulas.  For this minimal \(d^2\)-outcome IC frame, the unbiased linear dual is unique.  Relative to the corresponding SIC spectral benchmark, the bound overhead is
\begin{equation}
 \frac{d/\lambda(\phi)}{d+1}=\eta(\phi)^{-1}.
 \label{eq:shadow-overhead}
\end{equation}
Unbiasedness gives \(\E_\rho\widehat O_u=\Tr(O\rho)\), and \(\operatorname{Var}_\rho(\widehat O_u)\le\E_\rho[\widehat o_u^2]\).  Hence a median-of-means estimate of \(K\) observables from independent outcomes achieves additive error \(\epsilon\) for all of them with failure probability at most \(\delta\) using the universal bound
\begin{equation}
 N=O\!\left[
 \frac{d+1}{\eta(\phi)\epsilon^2}
 \max_{1\le j\le K}\Tr(O_{j,0}^2)
 \log\!\frac K\delta
 \right].
 \label{eq:shadow-samples}
\end{equation}
The full spectrum therefore determines the exact finite-sample canonical MSE at \(I/d\), whereas its lower edge controls inverse amplification, worst-direction Fisher efficiency, and universal canonical-shadow bounds.

\section{Completeness without stability}
\label{sec:separation}

\subsection{Generic does not mean stable}

The IC condition is generic, but the inverse problem can still have a
heavy lower spectral tail.

\begin{theorem}[Haar-generic IC and divergent inverse stability]
\label{thm:Haar}
Let \(\phi\) be Haar distributed on the unit sphere of \(\C^d\).  Then
\begin{equation}
 \Prob[\chi_\phi(u)\ne0\text{ for every }u\ne0]=1.
 \label{eq:Haar-IC}
\end{equation}
For each fixed nonidentity displacement \(u\),
\begin{equation}
 \E_\phi|\chi_\phi(u)|^2=\frac1{d+1}.
 \label{eq:Haar-coordinate}
\end{equation}
Equivalently, by the spectral correspondence
in Eq.~\eqref{eq:spectrum-multiset}, each fixed nonidentity
projector-Gram eigenvalue has mean
\begin{equation}
 \E_\phi[\lambda_v]=\frac{d}{d+1},
 \label{eq:Haar-eigenvalue-mean}
\end{equation}
which is exactly the WH-SIC nonidentity eigenvalue.  Nevertheless,
\begin{equation}
 \E_\phi[\lambda(\phi)^{-1}]=\infty.
 \label{eq:Haar-divergence}
\end{equation}
\end{theorem}

\begin{proof}[Proof sketch]
For fixed \(u\ne0\), the real polynomial
\(|\bra\phi D_u\ket\phi|^2\) is not identically zero, so its zero set
has Haar measure zero.  Since there are only finitely many nonidentity
displacements, a finite union proves Eq.~\eqref{eq:Haar-IC}.
This recovers, in the present Haar-measure formulation, the generic
informational-completeness result for finite Gabor POVMs in
Ref.~\cite{Goldberger2022}. 

The Haar second-moment identity
\begin{equation}
 \E_\phi|\bra\phi A\ket\phi|^2
 =
 \frac{\Tr(AA^\dagger)+|\Tr A|^2}{d(d+1)}
 \label{eq:Haar-second}
\end{equation}
gives Eq.~\eqref{eq:Haar-coordinate} for every nonidentity displacement,
since \(D_u\) is unitary and traceless for \(u\ne0\).
The spectral identity
\(\lambda_v=d|\chi_\phi(u)|^2\), up to the symplectic relabeling in
Eq.~\eqref{eq:spectrum-explicit}, then gives
Eq.~\eqref{eq:Haar-eigenvalue-mean}.

For the divergence, write \(p_j=|\phi_j|^2\), which is uniform on the
probability simplex.  If \(d\) is even, set
\begin{align*}
 B&=\sum_{j\,\mathrm{even}}p_j
 \sim\operatorname{Beta}(d/2,d/2),\\
 Y&=2B-1=\bra\phi Z^{d/2}\ket\phi.
\end{align*}
The density of \(B\) is continuous and strictly positive at \(B=1/2\),
so the density of \(Y\) is positive at \(Y=0\), and therefore
\(\E|Y|^{-2}=\infty\).

If \(d\) is odd,
\[
 \bra\phi Z\ket\phi=\sum_jp_j\omega^j
\]
is the linear image of the simplex onto the regular \(d\)-gon.  Its
two-dimensional density is bounded below near the interior point zero,
and
\[
 \int_0^\epsilon r^{-2}r\,dr=\infty.
\]
In both cases,
\[
 \lambda(\phi)
 \le d|\bra\phi D\ket\phi|^2
\]
for the displacement \(D\) used above, which proves
Eq.~\eqref{eq:Haar-divergence}.  Appendix~\ref{app:Haar} supplies the
constant-rank/coarea details for odd \(d\).
\end{proof}

Equations~\eqref{eq:Haar-coordinate} and
\eqref{eq:Haar-eigenvalue-mean} describe a fixed spectral direction:
on average, its Gram eigenvalue sits exactly at the SIC value.
Stability, however, is controlled by the minimum over all nonidentity
directions.  Haar fiducials are IC almost surely, yet rare realizations
with a near-zero minimum eigenvalue make
\(\lambda(\phi)^{-1}\) have infinite expectation.  Thus generic
completeness, and even SIC-level mean behavior in every fixed direction,
do not provide a finite-mean guarantee for worst-direction conditioning.
Accordingly, the Haar stability benchmark below is summarized by
quantiles of \(\eta\), rather than by the mean of
\(\lambda(\phi)^{-1}\). 

\subsection{BIC is not a stability level}

A BIC-POVM is a general notion and does not require group covariance.
In dimension \(d\), it has effects \(P_j/d\), where the \(P_j\) are
\(d^2\) rank-one projectors forming a basis of \(M_d(\C)\) and satisfying
\(\sum_j P_j=dI\) \cite{Farkas2026}.  Ref.~\cite{Farkas2026} proves
that such measurements exist in every dimension and, for this purpose,
gives an explicit Weyl--Heisenberg-covariant construction.  For a
complete rank-one WH orbit, covariance already implies
\(\sum_u P_u=dI\).  Hence within the WH subclass considered here,
\begin{equation}
 \text{WH-BIC}\quad\Longleftrightarrow\quad\text{WH-IC}.
 \label{eq:BIC-IC}
\end{equation}
Thus BIC fixes a structural completeness property, whereas spectral
stability asks how far the corresponding Gram spectrum stays from
singularity.  In particular, WH-BIC only requires the nonidentity
projector-Gram eigenvalues to be nonzero; it does not impose a
dimension-independent lower bound on them.

The all-dimensional WH construction used in
Ref.~\cite{Farkas2026} is based on the truncated geometric fiducial
\begin{equation}
 \ket{\psi_\alpha}
 =
 c_d\sum_{j=0}^{d-1}\alpha^j\ket j,
 \qquad
 c_d=
 \sqrt{\frac{1-|\alpha|^2}{1-|\alpha|^{2d}}},
 \quad
 0<|\alpha|<1.
 \label{eq:geometric-state}
\end{equation}
This ansatz goes back to D'Ariano \emph{et al.}, together with a WH
nonzero-coefficient criterion
\cite[Eqs.~(27) and (28)]{DAriano2004}.
Farkas \emph{et al.} observed that the original IC assertion requires
an amendment and supplied an explicit sufficient parameter region
\cite{Farkas2026}.  The following proposition sharpens this picture by
giving the complete phase classification and, at the same time,
quantifying the resulting lower spectral scale.

\begin{proposition}[Geometric zeros and spectral decay]
\label{prop:geometric}
Write
\[
 \alpha=\rho e^{i\theta},
 \qquad
 0<\rho<1.
\]
The orbit of \(\psi_\alpha\) is IC in every odd dimension.  In even
dimension it is IC if and only if
\begin{equation}
 d\theta\notin\pi\Z.
 \label{eq:geometric-criterion}
\end{equation}
In even dimension with \(d\theta\in\pi\Z\), all zeros have translation label \(a=d/2\),
their modulation labels obey
\begin{equation}
 (-1)^b=-e^{-id\theta},
 \label{eq:geometric-parity}
\end{equation}
and there are exactly \(d/2\) zero projector-Gram eigenvalues.  In
every dimension,
\begin{equation}
 \eta(\psi_\alpha)
 \le
 \frac{4(d+1)\rho^{2\lfloor d/2\rfloor}}
 {(1-\rho^{2d})^2}.
 \label{eq:geometric-bound}
\end{equation}
\end{proposition}

The proof, including the comparison between displacement conventions,
is given in Appendix~\ref{app:geometric}.  The explicit sufficient
family of Ref.~\cite{Farkas2026} takes
\[
 \alpha=\rho e^{2\pi i t},
 \qquad
 \rho<\frac12,
\]
with arbitrary \(t\) in odd dimension and
\(t\notin(2d)^{-1}\Z\) in even dimension.  Equation~\eqref{eq:geometric-bound}
therefore yields the parameter-independent envelope
\begin{equation}
 \eta
 <
 \frac{4(d+1)2^{-2\lfloor d/2\rfloor}}
 {(1-2^{-2d})^2}
 =
 O(d\,2^{-d}).
 \label{eq:BIC-envelope}
\end{equation}
Hence the SIC-normalized floor of the explicit WH family used to
establish BIC existence in every dimension is bounded above by an
exponentially decaying envelope.  This does not alter the ideal device-independent
randomness result of Ref.~\cite{Farkas2026}, which concerns exact
certification from BIC structure at the maximal Bell value.

The same phase classification also resolves a previously used
numerical choice.  For
\(\alpha=(1+i)/2\), considered in Ref.~\cite{Singal2022},
Proposition~\ref{prop:geometric} gives exactly \(d/2\) zero
projector-Gram eigenvalues whenever \(4\mid d\), while the lower
spectral scale decays exponentially in the remaining dimensions.
Thus, when \(4\mid d\), this particular WH POVM is neither IC nor
extremal because its \(d^2\) rank-one effects are linearly dependent.

\subsection{Polynomial stability in every cyclic dimension}

Exponential ill conditioning is not unavoidable even in arbitrary cyclic dimensions.  The explicit parity-dependent family below remains IC in every integer dimension and loses stability only polynomially.  It therefore occupies an intermediate regime between the geometric family from the explicit sufficient IC region and the uniformly stable finite-field constructions developed below.

\begin{proposition}[Explicit polynomially stable cyclic WH-IC measurements]
\label{prop:universal-cyclic}
For every integer \(d\ge2\), define a normalized cyclic WH fiducial by the following formula, with \(\zeta=e^{2\pi i/(d+1)}\) in the even case:
\begin{equation}
 \ket{\phi_d^{\rm cyc}}=
 \begin{cases}
 \displaystyle\frac1{\sqrt{d-1}}\sum_{x=1}^{d-1}\ket x,
 &d\text{ odd},\\[3mm]
 \displaystyle\frac1{\sqrt{d+3}}
 \left(2\ket0+\sum_{x=1}^{d-1}\zeta^x\ket x\right),
 &d\text{ even}
 \end{cases}.
 \label{eq:universal-cyclic-state}
\end{equation}
Its orbit is IC.  In odd dimension its exact spectral floor is
\begin{equation}
 \lambda(\phi_d^{\rm cyc})
 =\frac{2d}{(d-1)^2}\left(1-\cos\frac\pi d\right)
 =\Theta(d^{-3}).
 \label{eq:universal-odd-floor}
\end{equation}
In even dimension,
\begin{equation}
 \frac{16}{d(d+1)^2(d+3)^2}
 \le\lambda(\phi_d^{\rm cyc})
 \le\frac{32\pi^4}{d^5},
 \label{eq:universal-even-bounds}
\end{equation}
so its floor is \(\Theta(d^{-5})\).  In particular, the parity-defined family obeys the all-dimensional bound
\begin{equation}
 \eta(\phi_d^{\rm cyc})
 \ge\frac{128}{75}\,d^{-5}.
 \label{eq:universal-eta-bound}
\end{equation}
\end{proposition}

Appendix~\ref{app:obstructions} derives the complete odd spectrum and exact displacement-labelled even spectral values directly from the ambiguity function.  WH twirling and IC imply that Proposition~\ref{prop:universal-cyclic} gives, in every integer dimension, an explicit minimal equal-weight rank-one POVM with exactly \(d^2\) outcomes.  Within this complete rank-one WH setting it is also BIC and can instantiate the ideal BIC-based device-independent randomness protocol of Ref.~\cite{Farkas2026}.  Combined with Corollary~\ref{cor:shadow}, the same construction gives an explicit all-dimensional canonical IC-shadow family whose worst-direction spectral factor \(d/\lambda(\phi_d^{\rm cyc})\) is \(O(d^4)\) in odd dimensions and \(O(d^6)\) in even dimensions.  The first consequence concerns the ideal maximal-violation protocol, while the second is a polynomial spectral guarantee for the canonical estimator.

\section{Uniformly stable finite-field constructions}
\label{sec:construction}

We use two distinct phase-space structures: the cyclic WH group
\(\Z_d^2\), defined for every integer dimension \(d\), and the
finite-field WH group \(\F_q^2\).  Proposition~\ref{prop:universal-cyclic}
applies to the former in every dimension.  The finite-field setting
supports two complementary uniformly stable families below: a dedicated
characteristic-two construction for \(q=2^m\), and the asymptotically
near-SIC balanced-Alltop construction for characteristic \(p\ge5\).

Let \(q=p^r\) be a prime power, let \(\F_q\) be the finite field, and let
\begin{equation}
 \psi(x)=\exp\!\left[\frac{2\pi i}{p}
 \Tr_{\F_q/\F_p}(x)\right]
 \label{eq:additive-character}
\end{equation}
be its canonical additive character.  On \(\cH_q=\C^{\F_q}\), define
\begin{equation}
 X_a\ket x=\ket{x+a},
 \qquad
 Z_b\ket x=\psi(bx)\ket x,
 \qquad a,b\in\F_q.
 \label{eq:field-WH}
\end{equation}
Set \(D_{a,b}=X_aZ_b\),
\(
\Pi_{a,b}=D_{a,b}\ketbra\phi\phi D_{a,b}^\dagger
\), and \(E_{a,b}=\Pi_{a,b}/q\).  Throughout this section,
\begin{equation}
 G_\phi=G_\phi^\Pi,
 \qquad
 G_\phi^\Pi[(a,b),(c,e)]
 =\Tr(\Pi_{a,b}\Pi_{c,e}),
 \label{eq:field-projector-Gram}
\end{equation}
is the \(q^2\times q^2\) projector Gram matrix, and \(\lambda(\phi)\) is its smallest nonidentity eigenvalue.  Thus every spectral interval stated below refers to nonidentity projector-Gram eigenvalues.  Finite-field WH twirling gives \(\sum_{a,b\in\F_q}\Pi_{a,b}=qI\).  Consequently, whenever the orbit is IC, its \(q^2\) rank-one projectors are linearly independent and define a minimal equal-weight BIC measurement.
Theorem~\ref{thm:WH-spectrum} extends after replacing cyclic characters by additive-field characters:
\begin{equation}
 \Spec(G_\phi)=
 \{q|\bra\phi X_aZ_b\ket\phi|^2:(a,b)\in\F_q^2\}.
 \label{eq:field-spectrum}
\end{equation}
Appendix~\ref{app:finite-field} fixes the phase-space labels.  When \(r>1\), this is the finite-field WH group rather than the cyclic group \(\Z_q^2\).
The Moyal identity and the SIC endpoint bound in Proposition~\ref{prop:SIC-endpoint} extend verbatim after replacing cyclic characters by additive-field characters.

\subsection{Characteristic two and multi-qubit dimensions}

Assume $q=2^m$ with $m\ge1$.  After choosing an $\mathbb F_2$-basis
of $\mathbb F_q$, we identify
$\mathcal H_q\simeq(\mathbb C^2)^{\otimes m}$, under which
\[
\frac{1}{\sqrt q}\sum_{x\in\mathbb F_q}\ket{x}
=
\ket{+}^{\otimes m},
\qquad
\ket{0}=\ket{0}^{\otimes m}.
\]
Hence the fiducial used below is simply a superposition of the two
product states $\ket{+}^{\otimes m}$ and $\ket{0}^{\otimes m}$.
We use finite-field notation to describe its Weyl--Heisenberg orbit.
Let
\[
\psi(x)=(-1)^{\Tr_{\mathbb F_q/\mathbb F_2}(x)}
\]
denote the canonical additive character, and set
\begin{align}
\ket{+_q}
&=\frac{1}{\sqrt q}\sum_{x\in\mathbb F_q}\ket{x},
\notag\\
\ket{\phi_{q,\theta}}
&=
\frac{\ket{+_q}+e^{i\theta}\ket{0}}
{\sqrt{N_{q,\theta}}},
\notag\\
N_{q,\theta}
&=2+\frac{2\cos\theta}{\sqrt q}.
\label{eq:char-two-state}
\end{align}
We choose
\begin{equation}
\theta_q=
\begin{cases}
5\pi/12, & q=2,\\
2\pi/3, & q=4,\\
3\pi/4, & q\ge8,
\end{cases}
\qquad
\ket{\phi_q^{(2)}}=\ket{\phi_{q,\theta_q}}.
\label{eq:char-two-phase}
\end{equation}

\begin{theorem}[Uniform stability in characteristic two]
\label{thm:char-two}
The finite-field WH orbit of \(\phi_q^{(2)}\) defines a minimal
equal-weight rank-one IC POVM with exactly \(q^2\) outcomes and hence a
BIC measurement.  For \(q\ge8\), set
\begin{equation}
 \mu_q=\frac{2}{(2-\sqrt{2/q})^2},
 \qquad
 \nu_q=
 \frac{q(1-\sqrt{2/q})^2}{(2-\sqrt{2/q})^2}.
 \label{eq:char-two-mu-nu}
\end{equation}
Its complete nonidentity projector-Gram spectrum is
\begin{equation}
 \Spec\!\left(
 G_{\phi_q^{(2)}}^\Pi\big|_{\boldsymbol{1}^{\perp}}
 \right)
 =
 \begin{cases}
  \{(2/3)^{(3)}\},&q=2,\\
  \{(4/9)^{(9)},(4/3)^{(6)}\},&q=4,\\
  \{(8/9)^{(63)}\},&q=8,\\
  \{\mu_q^{((q-1)^2)},\nu_q^{(2(q-1))}\},&q>8.
 \end{cases}
 \label{eq:char-two-spectrum}
\end{equation} 
Here $x^{(m)}$ denotes an eigenvalue $x$ with multiplicity $m$.
Thus, at $q=8$, all $63=q^2-1$ nonidentity eigenvalues are equal to $8/9$. 
Consequently,
\begin{equation}
 \lambda(\phi_q^{(2)})=
 \begin{cases}
  2/3,&q=2,\\
  4/9,&q=4,\\
  \displaystyle\frac{2}{(2-\sqrt{2/q})^2},&q\ge8,
 \end{cases}
 \label{eq:char-two-floor}
\end{equation}
and
\begin{align}
 \inf_{m\ge1}\lambda(\phi_{2^m}^{(2)})&=\frac49,
 \notag\\
 \inf_{m\ge1}\eta(\phi_{2^m}^{(2)})&=\frac12,
 &
 \eta(\phi_{2^m}^{(2)})&\longrightarrow\frac12.
 \label{eq:char-two-uniform}
\end{align}
Here \(\eta=(q+1)\lambda/q\); in particular,
\(\eta(\phi_2^{(2)})=\eta(\phi_8^{(2)})=1\) and
\(\eta(\phi_4^{(2)})=5/9\).  The \(q=2\) and \(q=8\) members attain
the finite-field WH SIC endpoint.
\end{theorem}

\begin{proof}
For \(D_{a,b}=X_aZ_b\), characteristic two gives
\begin{align}
 \bra{+_q}D_{a,b}\ket{+_q}&=\delta_{b,0},&
 \bra{+_q}D_{a,b}\ket0&=q^{-1/2},
 \notag\\
 \bra0D_{a,b}\ket{+_q}&=q^{-1/2}\psi(ab),&
 \bra0D_{a,b}\ket0&=\delta_{a,0}.
 \label{eq:char-two-matrix-elements}
\end{align}
It follows that
\begin{equation}
 \chi_{\phi_{q,\theta}}(a,b)
 =
 \frac{\delta_{b,0}+\delta_{a,0}
 +q^{-1/2}\left(e^{i\theta}+e^{-i\theta}\psi(ab)\right)}
 {N_{q,\theta}}.
 \label{eq:char-two-characteristic}
\end{equation}
By Eq.~\eqref{eq:field-spectrum}, the nonidentity projector-Gram
eigenvalues therefore belong to the three branches
\begin{align}
 \mu_{\rm ax}(q,\theta)
 &=\frac{q(1+2\cos\theta/\sqrt q)^2}{N_{q,\theta}^2},
 \label{eq:char-two-axis}\\
 \mu_0(q,\theta)
 &=\frac{4\cos^2\theta}{N_{q,\theta}^2},
 \label{eq:char-two-trace-zero}\\
 \mu_1(q,\theta)
 &=\frac{4\sin^2\theta}{N_{q,\theta}^2}.
 \label{eq:char-two-trace-one}
\end{align}
The axis branch has multiplicity \(2(q-1)\).  For each \(a\ne0\),
the map \(b\mapsto\Tr_{\F_q/\F_2}(ab)\) is a nonzero
\(\F_2\)-linear functional.  The mixed trace-zero and trace-one
branches consequently have multiplicities \((q-1)(q/2-1)\) and
\((q-1)q/2\), respectively; the former is absent for \(q=2\).

For \(q=2\), substitution of \(\theta=5\pi/12\) makes both present
branches equal to \(2/3\).  For \(q=4\) and \(\theta=2\pi/3\),
\begin{equation}
 \mu_{\rm ax}=\mu_0=\frac49,
 \qquad
 \mu_1=\frac43.
\end{equation}
For \(q\ge8\), taking \(\theta=3\pi/4\) gives
\begin{equation}
 N_{q,\theta}=2-\sqrt{2/q},
 \qquad
 \mu_0=\mu_1=\mu_q,
 \qquad
 \mu_{\rm ax}=\nu_q.
\end{equation}
Moreover,
\begin{equation}
 \frac{\nu_q}{\mu_q}
 =\frac{(\sqrt q-\sqrt2)^2}{2}\ge1,
 \label{eq:char-two-comparison}
\end{equation}
with equality exactly at \(q=8\).  This proves the complete spectrum and
the floor, including the coalescence of all \(63\) nonidentity
eigenvalues at \(q=8\).

All three branches are strictly positive when present, so the orbit is
IC.  Finite-field WH twirling makes its \(q^2\) rank-one effects an
equal-weight POVM; IC makes them linearly independent and hence minimal,
and Eq.~\eqref{eq:BIC-IC} gives the BIC conclusion.  Finally,
\(\mu_q>1/2\) for finite \(q\ge8\) and \(\mu_q\to1/2\).  The displayed
low-dimensional values prove Eq.~\eqref{eq:char-two-uniform}.  At
\(q=2,8\), every nonidentity eigenvalue equals \(q/(q+1)\), so
Proposition~\ref{prop:SIC-endpoint} gives the endpoint assertion.
\end{proof}

After choosing an \(\F_2\)-basis of \(\F_q\) and its trace-dual basis,
\(\cH_q\cong(\C^2)^{\otimes m}\), \(\ket{+_q}=\ket+^{\otimes m}\), and
the finite-field displacement operators become tensor-product Pauli
operators up to phases.  Theorem~\ref{thm:char-two} therefore applies to
all multi-qubit Hilbert-space dimensions; this is a Hilbert-space
statement and does not assert an efficient measurement circuit.  Relative
to the even-dimensional cyclic family, it replaces a \(\Theta(q^{-5})\)
floor by a uniform one, but uses the different phase space \(\F_q^2\).
It gives the SIC-scale inverse and canonical-shadow factor
\begin{equation}
 \frac{q}{\lambda(\phi_q^{(2)})}\le\frac94q,
 \label{eq:char-two-operational}
\end{equation}
and its complete spectrum determines the canonical tomography MSE in
Corollary~\ref{cor:linear-MSE}.  As a BIC measurement, it can also
instantiate the ideal protocol of Ref.~\cite{Farkas2026}.  On the other
hand, Eq.~\eqref{eq:char-two-comparison} grows as \(q/2\), so the family
is uniformly stable but not asymptotically spectrally flat.  The
balanced-Alltop family below retains the stronger near-SIC conclusion
where it applies.

\subsection{A flat profile with a zero axis}

For the remainder of this section, let \(q=p^r\) have characteristic
\(p\ge5\).  The construction has a simple spectral picture.  Away from
one missing phase-space axis, the unperturbed Alltop ambiguity profile is
flat at magnitude \(q^{-1/2}\), asymptotically the SIC scale.  A
coordinate spike lifts the missing axis while only slightly distorting
the flat bulk.  The repaired-axis amplitude is proportional to
\(t^2+2t/\sqrt q\), whose quadratic term dominates at the balanced scale,
whereas the bulk distortion is linear in \(t\).  Balancing the two
selects a spike size of order \(q^{-1/4}\).  This cubic mechanism excludes
characteristics two and three; characteristic two is handled separately
by Theorem~\ref{thm:char-two}, while a uniformly stable finite-field
construction in characteristic three remains open.

Define the cubic Alltop state \cite{Alltop1980,Klappenecker2004,Hall2013,Goldberger2022}
\begin{equation}
 \ket{A_q}=\frac1{\sqrt q}
 \sum_{x\in\F_q}\psi(x^3)\ket x.
 \label{eq:Alltop}
\end{equation}

\begin{lemma}[Alltop ambiguity profile]
\label{lem:Alltop-profile}
Its characteristic function satisfies
\begin{equation}
 |\chi_{A_q}(a,b)|=
 \begin{cases}
 1,&(a,b)=(0,0),\\
 0,&a=0,\ b\ne0,\\
 q^{-1/2},&a\ne0.
 \end{cases}
 \label{eq:Alltop-profile}
\end{equation}
Consequently,
\begin{equation}
 \Spec(G_{A_q})=
 \{q^{(1)},1^{(q(q-1))},0^{(q-1)}\}.
 \label{eq:Alltop-spectrum}
\end{equation}
\end{lemma}

\begin{proof}
Direct substitution gives
\begin{equation}
 \chi_{A_q}(a,b)=\frac{\psi(-a^3)}q
 \sum_{x\in\F_q}\psi[-3ax^2+(b-3a^2)x].
 \label{eq:Alltop-Gauss}
\end{equation}
For \(a=0\), additive-character orthogonality gives the first two cases.  For \(a\ne0\), the quadratic coefficient is nonzero because \(p\ge5\); the standard finite-field quadratic Gauss identity gives modulus \(\sqrt q\) \cite[Chap.~5]{LidlNiederreiter1997}.  Equation~\eqref{eq:field-spectrum} gives Eq.~\eqref{eq:Alltop-spectrum}.
\end{proof}

The Alltop profile is incomplete, but all of its nonzero directions are already perfectly flat.  The next lemma isolates the zero-set repair mechanism and applies to any phase state with this profile, not only to the cubic example.

\begin{lemma}[Spectral repair of a zero axis]
\label{lem:repair}
Let
\begin{equation}
 \ket h=\frac1{\sqrt q}\sum_{x\in\F_q}\nu(x)\ket x,
 \qquad |\nu(x)|=1,\quad \nu(0)=1,
 \label{eq:flat-phase-state}
\end{equation}
and assume \(|\chi_h(a,b)|=q^{-1/2}\) for every \(a\ne0\).  For \(0<t<1/2\), set
\begin{equation}
 \ket{h_t}=\frac{\ket h+t\ket0}{\sqrt{N_{q,t}}},
 \qquad
 N_{q,t}=1+t^2+\frac{2t}{\sqrt q}.
 \label{eq:general-spike}
\end{equation}
Then, for \(b\ne0\),
\begin{equation}
 \chi_{h_t}(0,b)=
 \frac{t^2+2t/\sqrt q}{N_{q,t}},
 \label{eq:repair-axis}
\end{equation}
whereas for \(a\ne0\),
\begin{equation}
 \frac{1-2t}{\sqrt q\,N_{q,t}}
 \le |\chi_{h_t}(a,b)|
 \le\frac{1+2t}{\sqrt q\,N_{q,t}}.
 \label{eq:repair-off-axis}
\end{equation}
In particular, the orbit is IC and
\begin{equation}
 \lambda(h_t)\ge
 B_q(t):=
 \frac{\min\{(1-2t)^2,(\sqrt q\,t^2+2t)^2\}}
 {N_{q,t}^2}.
 \label{eq:repair-bound}
\end{equation}
\end{lemma}

\begin{proof}
Flat coordinate probabilities imply \(\chi_h(0,b)=0\) for \(b\ne0\).  The normalization in Eq.~\eqref{eq:flat-phase-state} gives
\(
\bra hX_aZ_b\ket0=q^{-1/2}\overline{\nu(a)}
\)
and
\(
\bra0X_aZ_b\ket h=q^{-1/2}\nu(-a)\psi(-ab),
\)
both of modulus \(q^{-1/2}\).  Moreover, \(\bra0X_aZ_b\ket0=\delta_{a,0}\).  Expanding the numerator of \(\chi_{h_t}\) gives Eq.~\eqref{eq:repair-axis}; the triangle and reverse-triangle inequalities give Eq.~\eqref{eq:repair-off-axis}.  Multiplication by \(q\) and Eq.~\eqref{eq:field-spectrum} yield Eq.~\eqref{eq:repair-bound}.
\end{proof}

Equation~\eqref{eq:repair-bound} makes this repair mechanism quantitative; the next step balances its two competing spectral scales exactly.

\subsection{Balanced perturbation and exact spectral floor}

Specialize Lemma~\ref{lem:repair} to \(h=A_q\).  The certified bound is uniquely maximized at the crossing of its two terms.  Indeed,
\begin{align}
 f_q(t)&:=\frac{1-2t}{N_{q,t}}
 &&\text{strictly decreases},\notag\\
 g_q(t)&:=\frac{\sqrt q\,t^2+2t}{N_{q,t}}
 &&\text{strictly increases}.
 \label{eq:monotone-amplitudes}
\end{align}
Direct differentiation gives
\begin{align}
 f_q'(t)&=-\frac{2(1+q^{-1/2}+t-t^2)}{N_{q,t}^2}<0,
 \notag\\
 g_q'(t)&=\frac{2(\sqrt q\,t+1)}{N_{q,t}^2}>0.
\end{align}
Their unique crossing is
\begin{equation}
 \sqrt q\,t^2+4t=1,
 \qquad
 t=t_q:=\frac1{\sqrt{4+\sqrt q}+2}.
 \label{eq:balanced-t}
\end{equation}
This optimizes the analytic lower bound \(B_q(t)\); it does not assert that \(t_q\) is the exact finite-\(q\) optimizer of the true minimum within the one-spike family, or a global max--min fiducial.

Set
\begin{equation}
 \ket{\phi_q}=\frac{\ket{A_q}+t_q\ket0}{\sqrt{N_q}},
 \qquad
 N_q=1+t_q^2+\frac{2t_q}{\sqrt q},
 \label{eq:balanced-state}
\end{equation}
and define
\begin{equation}
 L_q=\frac{(1-2t_q)^2}{N_q^2},
 \qquad
 U_q=\frac{(1+2t_q)^2}{N_q^2}.
 \label{eq:L-U}
\end{equation}

\begin{theorem}[Balanced-Alltop spectral interval and exact floor]
\label{thm:main-Alltop}
For every prime power \(q=p^r\) of characteristic \(p\ge5\),
\begin{equation}
 \lambda(\phi_q)=L_q.
 \label{eq:exact-L}
\end{equation}
Every nonidentity projector-Gram eigenvalue lies in \([L_q,U_q]\), and \(L_q\) has multiplicity at least \(q-1\).  Moreover, \(L_q\) is strictly increasing as a function of real \(q>0\), and therefore
\begin{equation}
 L_q\ge L_5=0.197863708777\ldots.
 \label{eq:uniform-L}
\end{equation}
\end{theorem}

\begin{proof}
At the balanced value, Eqs.~\eqref{eq:repair-axis} and \eqref{eq:balanced-t} give
\begin{equation}
 q|\chi_{\phi_q}(0,b)|^2=L_q,
 \qquad b\ne0.
 \label{eq:axis-attains}
\end{equation}
Equation~\eqref{eq:repair-off-axis} places every off-axis nonidentity eigenvalue in \([L_q,U_q]\).  Hence the axis attains the global minimum and gives its multiplicity.

For monotonicity, put \(s=\sqrt q\) and \(a_q=1-2t_q\).  The balance equation implies
\begin{equation}
 N_q=1+\frac{a_q}{s},
 \qquad
 L_q=\left(\frac{a_qs}{a_q+s}\right)^2.
 \label{eq:L-monotone-form}
\end{equation}
Both \(a_q\) and \(s\) increase with \(q\), while \(as/(a+s)\) increases in each positive argument.  The smallest allowed dimension is \(q=5\).
\end{proof}

The lower bound \(L_5\) makes the family uniformly spectrally stable.  The simultaneous confinement to \([L_q,U_q]\) is stronger: because \(U_q/L_q\to1\), every nonidentity direction becomes asymptotically equivalent.  The resulting \(\eta(\phi_q)\to1\) already reaches the SIC-normalized endpoint asymptotically; the next subsection shows that the attained floor also approaches the global finite-field WH max--min optimum without assuming SIC existence.  Appendix~\ref{app:finite-field} gives a closed formula for every off-axis eigenvalue, including its finite-field Gauss phase.

\subsection{Asymptotic optimality for the finite-field WH max--min metric}

As \(q\to\infty\) through prime powers of characteristic \(p\ge5\), writing \(x=q^{-1/4}\) in the exact formulas gives
\begin{align}
 t_q&=q^{-1/4}-2q^{-1/2}+O(q^{-3/4}),
 \label{eq:t-asymptotic}\\
 L_q&=1-4q^{-1/4}+10q^{-1/2}+O(q^{-3/4}),
 \label{eq:L-asymptotic}\\
 U_q&=1+4q^{-1/4}-6q^{-1/2}+O(q^{-3/4}).
 \label{eq:U-asymptotic}
\end{align}
Thus
\begin{equation}
 \eta(\phi_q)=\frac{q+1}{q}L_q\longrightarrow1.
 \label{eq:eta-to-one}
\end{equation}
More intrinsically, let
\begin{equation}
 \Lambda_q^\star=\max_{\norm\phi=1}\lambda(\phi)
 \label{eq:global-optimum}
\end{equation}
for the finite-field WH group.  Proposition~\ref{prop:SIC-endpoint} and Theorem~\ref{thm:main-Alltop} give
\begin{align}
 L_q&\le\Lambda_q^\star\le\frac q{q+1},
 \notag\\[-1mm]
 \frac{q+1}{q}L_q
 &\le\frac{L_q}{\Lambda_q^\star}\le1,
 \qquad
\frac{L_q}{\Lambda_q^\star}\longrightarrow1.
 \label{eq:global-squeeze}
\end{align}
No SIC-existence assumption enters this squeeze.  It also gives the quantitative gap
\begin{equation}
 0\le\Lambda_q^\star-L_q
 \le\frac q{q+1}-L_q
 =4q^{-1/4}+O(q^{-1/2}).
 \label{eq:global-gap}
\end{equation}

At finite \(q\), \(\phi_q\) is not a SIC.  To see this directly, set \(s=\sqrt q\) and \(a=a_q\in(0,1)\).  Equation~\eqref{eq:L-monotone-form} gives
\begin{equation}
 L_q=\frac{a^2s^2}{(a+s)^2}
 <\frac{s^2}{s^2+1}=\frac q{q+1},
 \label{eq:not-SIC}
\end{equation}
where the strict inequality is equivalent to \(a^2s<s+2a\), which follows from \(a^2<1\).

The same interval controls pairwise orbit geometry.  For distinct \(u,v\),
\begin{equation}
 \frac{L_q}{q}\le\Tr(\Pi_u\Pi_v)\le\frac{U_q}{q},
 \label{eq:pairwise-bounds}
\end{equation}
and hence
\begin{equation}
 \max_{u\ne v}\abs{(q+1)\Tr(\Pi_u\Pi_v)-1}
 =O(q^{-1/4}).
 \label{eq:approx-equiangular}
\end{equation}
The traceless-sector condition number satisfies
\begin{equation}
 \kappa_{\mathrm{tr}}(G_{\phi_q})
 \le\frac{U_q}{L_q}
 =\left(\frac{1+2t_q}{1-2t_q}\right)^2
 =1+8q^{-1/4}+O(q^{-1/2}).
 \label{eq:condition-number}
\end{equation}
Thus the nonidentity spectrum becomes asymptotically isotropic.  More operationally, every traceless Hermitian local direction at \(I/q\) obeys
\begin{equation}
 \frac{L_q}{q}
 \le \frac{I_C}{I_Q}
 \le \frac{U_q}{q}.
 \label{eq:Alltop-Fisher}
\end{equation}
After normalization by the SIC directional value \(1/(q+1)\), both endpoints tend to one, and their ratio tends to one.  Uniform stability here is SIC-normalized: the physical scaled-frame minimum is \(L_q/q\), so \(\norm{\cM_{\phi_q}^{-1}}_{2\to2}=q/L_q=\Theta(q)\), the same unavoidable dimensional scaling as the SIC value \(d+1\).  The gain is a dimension-independent relative factor and vanishing traceless anisotropy.  Finally, Corollary~\ref{cor:shadow} gives
\begin{equation}
 \sup_\rho\E_\rho[\widehat o^2]
 \le\frac q{L_q}\Tr(O_0^2)
 \le\frac q{L_5}\Tr(O_0^2),
 \label{eq:Alltop-shadow}
\end{equation}
and the ratio of this universal spectral bound to the SIC bound tends to one.

\section{Numerical benchmarks and reproducibility}
\label{sec:numerics}

Figure~\ref{fig:story} uses the common projector-Gram normalization \(\eta=(d+1)\lambda/d\).  The conceptual panel separates pointwise IC, uniform stability, and the SIC endpoint; the data panels strictly separate cyclic \(\Z_d^2\) and finite-field \(\F_q^2\) phase spaces.

All WH characteristic coefficients in the cyclic benchmarks are evaluated by FFT-based cyclic correlation in \(O(d^2\log d)\) time.  Figure~\ref{fig:story}(b) uses 2000 Haar fiducials in each dimension \(2\le d\le50\), and Fig.~\ref{fig:operational-MSE} uses 4000 in each displayed prime dimension.  The implementation is independently validated against direct Gram diagonalization, the Moyal identity, the analytic geometric zero labels, and the closed Alltop spectra.  Reproducibility details are given in Appendix~\ref{app:numerics}.

The BIC paper selects no preferred geometric parameter, so the pink choice is only a reproducible representative; the dashed curve is an exponentially decaying upper envelope for the explicit sufficient parameter region.  In contrast, the balanced-Alltop formula rises from \(\eta_5\approx0.2374\) to \(\eta_{9973}\approx0.6873\); the theorem proves convergence to one at the slow rate \(q^{-1/4}\).

Figure~\ref{fig:operational-MSE} evaluates Corollary~\ref{cor:linear-MSE} analytically from the complete spectrum; the Haar band describes variation across fiducials rather than Monte Carlo measurement-shot error.  This is a fixed-state comparison for canonical unbiased inversion at
\(I/q\), rather than a state-uniform or estimator-optimal statement.
No numerical optimization enters Theorem~\ref{thm:main-Alltop}.
The characteristic-two spectrum in Theorem~\ref{thm:char-two} is likewise
an exact analytic result and is not inferred from the plotted data.

\section{Discussion}
\label{sec:discussion}

Minimal WH measurement design has a natural spectral objective,
\begin{equation}
 \max_{\norm\phi=1}\min_{u\ne0}|\chi_\phi(u)|^2.
 \label{eq:discussion-max-min}
\end{equation}
The Moyal identity fixes the average nonidentity spectral weight, and a
WH SIC is exactly the flat-spectrum maximizer.  Balanced Alltop reaches
this endpoint asymptotically in a stronger sense than convergence of its
minimum alone: the full nonidentity interval collapses,
\(U_q/L_q\to1\), while
\(\lambda(\phi_q)/\Lambda_q^\star\to1\).  It therefore gives an explicit
asymptotically optimal solution to the finite-field WH worst-direction
problem without assuming that an exact SIC exists.  This
spectral objective complements approximate-SIC criteria based on
pairwise coherence: those criteria control individual orbit overlaps,
whereas the Gram edge identifies the least resolved operator direction.

The balanced perturbation also suggests a simple design principle.
The unperturbed Alltop state already has an almost ideal flat ambiguity
profile except for a missing phase-space axis; a single-coordinate
perturbation repairs that zero set while only weakly distorting the
remaining spectrum.  Lemma~\ref{lem:repair} isolates this spectral
zero-set repair mechanism in a form that may apply to other structured
phase states, with broader Alltop functions providing natural
candidates \cite{Hall2013}.  The competing scales reveal why the
construction is balanced at \(t_q\asymp q^{-1/4}\): the repaired-axis
amplitude is \(t^2+2t/\sqrt q\), with the quadratic term dominant at this
scale, while the distortion of the flat bulk is linear in \(t\).

The remaining boundaries are concrete.  The finite-field phase space
\(\F_q^2\) differs from \(\Z_q^2\) when \(q=p^r\) with \(r>1\), so the
prime-power results do not automatically extend the cyclic construction.
A uniformly stable finite-field construction in characteristic three
remains open.  The characteristic-two fiducial
covers every multi-qubit Hilbert space but by itself supplies no efficient
implementation circuit.  Balanced Alltop is not an exact SIC at finite
\(q\), and its chosen spike optimizes the certified analytic bound rather
than a proved finite-\(q\) global objective.  It is also natural to ask
whether overcomplete or augmented frames improve the canonical-estimator
guarantees \cite{Innocenti2023,Fischer2024}, and whether the present
single-frame spectrum predicts robustness of BIC-based certification away
from the ideal maximal-violation point.

\section{Conclusion}

Informational completeness says whether inversion exists; the weakest
projector-Gram eigenvalue says how well that inverse resolves its hardest
operator direction.  WH covariance makes this quantity explicit and
places the SIC at its max--min endpoint.

Explicit minimal measurements can then be organized by increasing
stability: polynomial floors in every integer dimension, a uniform floor
for every multi-qubit dimension, and, in finite-field dimensions of
characteristic at least five, a balanced-Alltop spectrum that becomes
isotropic and approaches the global finite-field WH max--min optimum
without assuming SIC existence.

This spectrum controls canonical inverse amplification and shadow bounds,
fixes the weakest local Fisher direction at \(I/d\), and determines the
exact finite-sample Hilbert--Schmidt MSE of canonical linear inversion
there.  Minimal IC is therefore the starting point; spectral design
determines whether the measurement remains statistically useful as the
dimension grows.

\begin{acknowledgments}
This work was supported by the National Natural Science Foundation of China under Grants No.~42330707 and No.~42530108, and by the Beijing Natural Science Foundation under Grant No.~Z220002.

OpenAI GPT-5.6 Sol was used as an assistive tool during exploratory
mathematical reasoning, including the exploration of candidate
constructions, derivation checking, numerical verification, and
manuscript editing.  All mathematical claims, proofs, and numerical
results reported here were checked by the authors, who take full
responsibility for the content of the manuscript.
\end{acknowledgments}

\appendix

\section{Labelled cyclic WH spectrum}
\label{app:spectrum}

\subsection{Route A: Fourier/circulant diagonalization}

Equation~\eqref{eq:projector-Gram} depends only on the phase-space difference \(v-u\), so \(G_\phi\) is a two-dimensional circulant matrix.  Substitution of a character gives
\begin{equation}
 G_\phi f_{m,n}
 =\left(\sum_{a,b}g_\phi(a,b)\omega^{ma+nb}\right)f_{m,n}.
 \label{eq:appendix-Fourier-route}
\end{equation}
The \(d^2\) characters are orthonormal and complete, while Eq.~\eqref{eq:ambiguity-self-Fourier} identifies the coefficient in parentheses as \(d|\chi_\phi(-n,m)|^2\).  This route makes the Fourier eigendirections immediate.

\subsection{Route B: synthesis-operator reduction}

Let \(T_\phi:\C^{\Z_d^2}\to\mathcal L(\C^d)\) be the synthesis operator
\begin{equation}
 T_\phi c=\sum_{a,b}c(a,b)\Pi_{a,b}.
 \label{eq:synthesis}
\end{equation}
Then \(G_\phi=T_\phi^\dagger T_\phi\).  The normalized displacements \(d^{-1/2}D_{p,q}\) form an orthonormal operator basis, and
\begin{equation}
 D_{a,b}D_{p,q}D_{a,b}^\dagger
 =\omega^{bp-aq}D_{p,q}.
 \label{eq:cyclic-conjugation}
\end{equation}
Expanding \(\Pi_{0,0}\) in this basis and inserting Eq.~\eqref{eq:phase-character} gives
\begin{align}
 T_\phi f_{m,n}
 &=\frac1d\sum_{a,b}\omega^{ma+nb}
 D_{a,b}\Pi_{0,0}D_{a,b}^\dagger
 \notag\\
 &=\overline{\chi_\phi(-n,m)}D_{-n,m}.
 \label{eq:synthesis-character}
\end{align}
The character sum selects exactly the displacement \((-n,m)\), including when its coefficient vanishes.  For distinct Fourier labels, Weyl orthogonality gives explicitly
\begin{equation}
 \left\langle T_\phi f_{m,n},
 T_\phi f_{m',n'}\right\rangle_{\HS}=0,
 \qquad (m,n)\ne(m',n'),
 \label{eq:synthesis-image-orthogonality}
\end{equation}
so \(T_\phi^\dagger T_\phi\) is diagonal in the character basis.  Therefore
\begin{equation}
 \langle f_{m,n},G_\phi f_{m,n}\rangle
 =\norm{T_\phi f_{m,n}}_{\HS}^2
 =d|\chi_\phi(-n,m)|^2,
 \label{eq:synthesis-norm}
\end{equation}
which proves Eq.~\eqref{eq:spectrum-explicit}.  Route A explains why convolution selects Fourier characters; Route B explains why the corresponding DFT eigenvalues reduce to the ambiguity intensities themselves, including when a coefficient vanishes.

The same operator expansion gives
\begin{align}
 \cS_\phi(D_{p,q})
 &=d|\chi_\phi(p,q)|^2D_{p,q},\notag\\
 \cM_\phi(D_{p,q})
 &=|\chi_\phi(p,q)|^2D_{p,q},
 \label{eq:channel-spectrum-app}
\end{align}
and setting \(A=\Pi_{0,0}\) in Parseval's identity gives the Moyal sum \eqref{eq:Moyal}.

\section{Canonical second moments}
\label{app:shadow}

Let \(A=\cM_\phi^{-1}(O_0)\).  Since \(\cM_\phi\) is self-adjoint, has a positive real spectrum under IC, and preserves Hermitian operators, its inverse is self-adjoint and Hermiticity-preserving.  Consequently,
\begin{equation}
 \widehat o_u
 =\Tr[O_0\cM_\phi^{-1}(\Pi_u)]
 =\Tr(A\Pi_u)\in\mathbb R.
 \label{eq:shadow-value-app}
\end{equation}
Using Eq.~\eqref{eq:frame-channel},
\begin{align}
 \sum_u\widehat o_u^2
 &=\langle A,\cS_\phi(A)\rangle_{\HS}
 =d\langle O_0,\cM_\phi^{-1}(O_0)\rangle_{\HS}
 \notag\\
 &\le\frac{d^2}{\lambda(\phi)}\Tr(O_0^2).
 \label{eq:shadow-sum-app}
\end{align}
For an arbitrary state, the outcome probability satisfies
\(
p_u=\Tr(E_u\rho)\le1/d.
\)
Multiplying Eq.~\eqref{eq:shadow-sum-app} by this upper bound proves Eq.~\eqref{eq:shadow-worst}.  At \(\rho=I/d\), every outcome has probability \(1/d^2\), yielding the equality in Eq.~\eqref{eq:shadow-mixed}.

If a nonidentity projector-Gram eigenvalue is \(\lambda_v\), the corresponding eigenvalue of \(\cM_\phi^{-1}\) is \(d/\lambda_v\).  Expanding \(O_0\) in the orthonormal Weyl operator basis therefore bounds the exact mixed-input second moment between \(\Tr(O_0^2)/U\) and \(\Tr(O_0^2)/L\).

\section{Haar inverse stability}
\label{app:Haar}

Only the odd-dimensional density argument in Theorem~\ref{thm:Haar} requires elaboration.  Let
\begin{equation}
 \Delta=\{p\in\mathbb R^d:p_j\ge0,\ \textstyle\sum_jp_j=1\}
 \label{eq:simplex}
\end{equation}
with its uniform \((d-1)\)-dimensional measure, and write \(v_j=(\Re\omega^j,\Im\omega^j)\).  On the affine hyperplane
\(
H=\{p:\sum_jp_j=1\}
\)
define
\begin{equation}
 Lp=\sum_jp_jv_j\in\mathbb R^2.
 \label{eq:polygon-map}
\end{equation}
The uniform point \(p_*=(1/d,\ldots,1/d)\) lies in the interior of \(\Delta\) and satisfies \(Lp_*=0\).  On the translation space \(H_0=\{x:\sum_jx_j=0\}\), the restriction of \(L\) has rank two because the regular polygon is not contained in an affine line.

For \(d=3\), \(L\) is an affine isomorphism from \(\Delta\) onto a triangle, so the density of \(Y=Lp\) is constant and positive near zero.  For odd \(d\ge5\), choose a bounded right inverse \(R:\mathbb R^2\to H_0\) and let \(K=\ker(L|_{H_0})\), of dimension \(d-3\).  Since \(p_*\) is an interior point, there exist \(\delta,\epsilon>0\) such that
\begin{equation}
 p_*+Ry+k\in\Delta
 \quad\text{whenever }|y|<\delta,
 \quad k\in K,\quad |k|<\epsilon.
 \label{eq:fiber-ball}
\end{equation}
The linear coarea formula gives the density
\begin{equation}
 f_Y(y)=\frac{C_d}{J_L}
 \mathcal H^{d-3}\!\left(\Delta\cap L^{-1}(y)\right),
 \label{eq:coarea-density}
\end{equation}
where \(J_L>0\) is the constant two-dimensional Jacobian and \(C_d\) normalizes the simplex measure.  Equation~\eqref{eq:fiber-ball} implies \(f_Y(y)\ge c_d>0\) for \(|y|<\delta\).  Hence
\begin{equation}
 \E|Y|^{-2}
 \ge c_d\int_{|y|<\delta}|y|^{-2}\,d^2y
 =2\pi c_d\int_0^\delta\frac{dr}{r}=\infty.
 \label{eq:odd-divergence}
\end{equation}
This numerical-shadow viewpoint is consistent with the density framework of Ref.~\cite{Dunkl2011}.  Since \(\lambda(\phi)\le d|Y|^2\), Eq.~\eqref{eq:Haar-divergence} follows.

For completeness, in even dimension
\begin{equation}
 B=\sum_{j\,\mathrm{even}}p_j
 \sim\operatorname{Beta}(d/2,d/2),
 \qquad
 \bra\phi Z^{d/2}\ket\phi=2B-1.
 \label{eq:Beta-app}
\end{equation}
The beta density is continuous and strictly positive at \(B=1/2\), which directly gives the same inverse-square divergence.

\section{Structured zeros and arbitrary-dimensional cyclic constructions}
\label{app:obstructions}

The exact spectrum also gives several useful structural diagnostics.  Write \(\ket\phi=\sum_xc_x\ket x\) and \(S=\{x:c_x\ne0\}\).

\begin{proposition}[Structured obstructions]
\label{prop:obstructions}
Each condition below forces a nonidentity zero of \(\chi_\phi\):
\begin{enumerate}
 \item \(d\) is even and every \(c_x\) is real;
 \item the coordinate probabilities are flat, \(|c_x|^2=1/d\);
 \item \(S-S\ne\Z_d\), in particular if
 \begin{equation}
  |S|<\frac{1+\sqrt{4d-3}}2;
  \label{eq:support-threshold}
 \end{equation}
 \item \(\phi\) is a stabilizer state for the same WH/Clifford structure.
\end{enumerate}
\end{proposition}

\begin{proof}
For item 1, pair \(x\) and \(x+d/2\) in
\(
\chi_\phi(d/2,b)=\sum_x\overline{c_{x+d/2}}c_x\omega^{bx};
\)
the two terms cancel for odd \(b\).  For item 2,
\(
\chi_\phi(0,b)=d^{-1}\sum_x\omega^{bx}=0
\)
when \(b\ne0\).  For item 3, if \(a\notin S-S\), every term of \(\chi_\phi(a,b)\) vanishes.  The sufficient threshold follows from \(|S-S|\le |S|(|S|-1)+1\).  Finally, a stabilizer state has unit-modulus characteristic coefficients on a size-\(d\) stabilizer subgroup.  Those terms exhaust the Moyal sum, so every coefficient outside the subgroup vanishes.
\end{proof}

The strict inequality in Eq.~\eqref{eq:support-threshold} is essential, since a difference set can saturate the counting bound.

\subsection{Parity constructions and polynomial stability}

We now prove Proposition~\ref{prop:universal-cyclic} directly from the ambiguity function.  For odd \(d\ge3\), take the missing coordinate to be zero.  The unnormalized cyclic correlations of \(c_0=0\) and \(c_x=1\) for \(x\ne0\) are
\begin{equation}
 A_{a,b}:=\sum_x\overline{c_{x+a}}c_x\omega^{bx}
 =\begin{cases}
 d-1,&a=b=0,\\
 -1,&a=0,\ b\ne0,\\
 d-2,&a\ne0,\ b=0,\\
 -(1+\omega^{-ab}),&a\ne0,\ b\ne0.
 \end{cases}
 \label{eq:odd-ambiguity}
\end{equation}
Division by the squared norm \(d-1\), followed by Eq.~\eqref{eq:spectrum-multiset}, gives \(\mu_{a,b}:=d|A_{a,b}|^2/(d-1)^2\) and the displacement-labelled spectral values
\begin{equation}
 \mu_{a,b}=\begin{cases}
 d/(d-1)^2,&a=0,\ b\ne0,\\
 d(d-2)^2/(d-1)^2,&a\ne0,\ b=0,\\
 d[2+2\cos(2\pi ab/d)]/(d-1)^2,
 &a\ne0,\ b\ne0.
 \end{cases}
 \label{eq:odd-spectrum}
\end{equation}
where the Fourier-character label of Theorem~\ref{thm:WH-spectrum} is \(\lambda_{m,n}=\mu_{-n,m}\).
Oddness prevents a zero in the last line, and \(a=1\), \(b=(d\pm1)/2\) attains its minimum.  Hence
\begin{equation}
 \lambda(\phi_d^{\rm cyc})
 =\frac{2d}{(d-1)^2}\left(1-\cos\frac\pi d\right)
 \sim\frac{\pi^2}{d^3}.
 \label{eq:odd-min}
\end{equation}
The inequality \(1-\cos(\pi/d)\ge2/d^2\) also gives \(\eta(\phi_d^{\rm cyc})\ge4d^{-3}\).  The case \(d=3\) is a SIC.

For even \(d\), let \(c_x=(1+\delta_{x0})\zeta^x\), with \(\zeta=e^{2\pi i/(d+1)}\).  Thus \(\|c\|^2=d+3\); this special coefficient at zero is essential.  Put \(\sigma=\zeta^{-1}=\zeta^d\) and \(r=\omega^b\).  Splitting the correlation at the cyclic wrap gives the exact formula
\begin{equation}
 \begin{aligned}
 A_{0,0}&=d+3,\\
 A_{0,b}&=3, &&b\ne0,\\
 A_{a,0}&=\zeta^{-a}[d-a+1+(a+1)\sigma],
 &&a\ne0,\\
 A_{a,b}&=\zeta^{-a}\left[
 1+\sigma r^{-a}+\frac{(1-\sigma)(1-r^{-a})}{1-r}
 \right], &&a,b\ne0.
 \end{aligned}
 \label{eq:even-ambiguity}
\end{equation}
Consequently, the exact displacement-labelled spectral values are
\begin{equation}
 \begin{aligned}
 \mu_{a,b}&:=d|\chi_{\phi_d^{\rm cyc}}(a,b)|^2
 =\frac{d}{(d+3)^2}|A_{a,b}|^2,\\
 \lambda_{m,n}&=\mu_{-n,m}.
 \end{aligned}
 \label{eq:even-labelled-spectrum}
\end{equation}
For the off-axis labels define
\begin{equation}
 h=\frac\pi{d+1},\qquad
 k=\frac{\pi b}{d},\qquad
 T=\frac{\pi ab}{d}.
\end{equation}
Removing an overall phase from the last line of Eq.~\eqref{eq:even-ambiguity} yields
\begin{equation}
 \begin{aligned}
 \mu_{a,b}&=\frac{4d}{(d+3)^2}
 \left\{[\cos(T+h)-\sin T\sin h]^2\right.\\
 &\hspace{25mm}\left.
 +\sin^2T\sin^2h\cot^2k\right\}.
 \end{aligned}
 \label{eq:even-offaxis-spectrum}
\end{equation}

We next bound this expression uniformly away from zero.  Assume first that even \(d\ge4\), reduce \(T\) modulo \(\pi\) to \(t\in[0,\pi)\), and write the two real components of half the unnormalized off-axis overlap, up to a common sign, as
\begin{equation}
 R=\cos t\cos h-2\sin t\sin h,
 \qquad I=-\sin t\sin h\cot k.
 \label{eq:even-RI}
\end{equation}
If \(b\ne d/2\) and \(\sin t\ge1/2\), then
\(
 |\cot k|\ge\tan(\pi/d)\ge2/d
\)
and \(\sin h\ge2/(d+1)\), so
\begin{equation}
 |A_{a,b}|=2\sqrt{R^2+I^2}
 \ge\frac4{d(d+1)}.
 \label{eq:even-overlap-lower}
\end{equation}
If \(\sin t<1/2\), then \(t<\pi/6\) or \(t>5\pi/6\).  In the first interval,
\begin{equation}
 R\ge\frac{\sqrt3}{2}\cos\frac\pi5-\sin\frac\pi5
 >\frac1{10},
\end{equation}
while in the second interval the two terms in \(R\) have the same sign and \(|R|>1/10\).  Equation~\eqref{eq:even-overlap-lower} follows again.  For the exceptional label \(b=d/2\), Eq.~\eqref{eq:even-offaxis-spectrum} gives \(|A_{a,b}|=2\cos h\) for even \(a\) and \(4\sin h\) for odd \(a\), both larger than the same bound.  Finally, the two axes obey \(|A_{0,b}|=3\) and
\(
 |A_{a,0}|\ge(a+1)\sin(2h)\ge8/(d+1).
\)
For \(d=2\), direct evaluation gives \(\lambda=8/25\).  Thus every even dimension satisfies
\begin{equation}
 \lambda(\phi_d^{\rm cyc})
 \ge\frac{16}{d(d+1)^2(d+3)^2}.
 \label{eq:even-lower-bound}
\end{equation}

This fifth-power scale is sharp for the even branch.  For even \(d\ge6\), take \(a=1\), \(b=d/2-2\), and set
\begin{equation}
 \epsilon=\frac{2\pi}{d},\qquad
 \delta=\epsilon-2h=\frac{2\pi}{d(d+1)}.
\end{equation}
Then \(t=k=\pi/2-\epsilon\), and Eq.~\eqref{eq:even-RI} becomes
\begin{equation}
 R=\sin(h+\delta)-\cos\epsilon\sin h,
 \qquad |I|=\sin\epsilon\sin h.
\end{equation}
Using \(1-\cos\epsilon\le\epsilon^2/2\) and the Lipschitz bound for sine gives \(|R|,|I|\le2\pi^2/d^2\).  Hence
\begin{equation}
 \lambda(\phi_d^{\rm cyc})
 \le\frac{32\pi^4}{d^5};
 \label{eq:even-upper-bound}
\end{equation}
for \(d=2,4\) the same bound follows trivially from \(\lambda\le d/(d+1)\).  Equations~\eqref{eq:even-lower-bound} and \eqref{eq:even-upper-bound} prove \(\lambda=\Theta(d^{-5})\) along the even dimensions.  Moreover,
\begin{equation}
 \eta(\phi_d^{\rm cyc})
 \ge\frac{16}{d^2(d+1)(d+3)^2}
 \ge\frac{128}{75}\,d^{-5},
\end{equation}
where the last step uses \(d+1\le3d/2\) and \(d+3\le5d/2\).  The odd bound is stronger, completing the proof of Proposition~\ref{prop:universal-cyclic}.

\section{Complete geometric-family analysis}
\label{app:geometric}

Let \(z=\rho^2\omega^b\).  Splitting the defining overlap at the cyclic wrap gives
\begin{align}
 \chi_\alpha(a,b)
 &=\frac{c_d^2}{1-z}\left[
 \overline\alpha^{\,a}(1-z^{d-a})
 +\overline\alpha^{\,a-d}z^{d-a}(1-z^a)
 \right]
 \label{eq:geometric-chi-one}\\
 &=\frac{c_d^2\overline\alpha^{\,a-d}}{1-z}
 \left[\overline\alpha^{\,d}-\rho^{2d}
 +(1-\overline\alpha^{\,d})z^{d-a}\right].
 \label{eq:geometric-chi-two}
\end{align}
If this coefficient vanishes, then
\begin{equation}
 z^{d-a}=\frac{\rho^{2d}-\overline\alpha^{\,d}}
 {1-\overline\alpha^{\,d}}.
 \label{eq:geometric-zero-equation}
\end{equation}
The modulus of the right-hand side is exactly \(\rho^d\), since
\begin{equation}
 |\rho^{2d}-\overline\alpha^{\,d}|
 =\rho^d|1-\alpha^d|
 =\rho^d|1-\overline\alpha^{\,d}|.
 \label{eq:geometric-modulus}
\end{equation}
The left-hand side has modulus \(\rho^{2(d-a)}\).  Because \(0<\rho<1\), equality forces \(a=d/2\), and hence \(d\) is even.  With \(h=d/2\), Eq.~\eqref{eq:geometric-chi-one} factors as
\begin{equation}
 \chi_\alpha(h,b)
 =c_d^2\frac{1-z^h}{1-z}\overline\alpha^{-h}
 (\overline\alpha^{\,d}+z^h),
 \qquad z^h=\rho^d(-1)^b.
 \label{eq:geometric-factor}
\end{equation}
The first factor cannot vanish because \(|z^h|=\rho^d<1\).  The last factor vanishes exactly when Eq.~\eqref{eq:geometric-criterion} fails, namely \(d\theta\in\pi\Z\), and Eq.~\eqref{eq:geometric-parity} holds.  One parity class contains exactly \(d/2\) labels.

For the stability bound, use the two nonwrapped sums directly.  For any \(a\),
\begin{equation}
 |\chi_\alpha(a,b)|
 \le\frac{\rho^a(1-\rho^{2(d-a)})
 +\rho^{d-a}(1-\rho^{2a})}{1-\rho^{2d}}.
 \label{eq:geometric-triangle}
\end{equation}
Choosing \(a=\lfloor d/2\rfloor\) bounds this by
\(
2\rho^{\lfloor d/2\rfloor}/(1-\rho^{2d}).
\)
Squaring and using Eq.~\eqref{eq:intro-eta} proves Eq.~\eqref{eq:geometric-bound}.

The convention of Ref.~\cite{DAriano2004} is
\begin{equation}
 U_{m,n}=\sum_k\omega^{km}\ket k\!\bra{k+n}
 =\omega^{-mn}X^{-n}Z^m,
 \label{eq:DPS-convention}
\end{equation}
so it differs from ours only by the symplectic relabeling \((a,b)=(-n,m)\) and a global phase.  The conventions of Refs.~\cite{Singal2022,Farkas2026} are likewise related by a phase-space relabeling.  Ranks, zero multiplicities, and the minimum spectrum are invariant under these changes.

As a low-dimensional check, \(d=4\) and \(\alpha=(1+i)/2\) have \(a=2\) and even \(b\) as zeros, giving two zero Gram eigenvalues exactly as Proposition~\ref{prop:geometric} predicts.

\section{Finite-field labels and complete Alltop spectrum}
\label{app:finite-field}

Let the symplectic form on \(\F_q^2\) be
\begin{equation}
 [(a,b),(c,e)]=bc-ae.
 \label{eq:symplectic-form}
\end{equation}
Conjugation obeys
\begin{equation}
 D_{a,b}D_{c,e}D_{a,b}^\dagger
 =\psi([(a,b),(c,e)])D_{c,e}.
 \label{eq:field-conjugation}
\end{equation}
For \(v\in\F_q^2\), define the normalized phase-space character
\begin{equation}
 f_v(u)=q^{-1}\psi([u,v]).
 \label{eq:field-Fourier-vector}
\end{equation}
Repeating the synthesis proof in Appendix~\ref{app:spectrum} gives
\begin{equation}
 G_\phi f_v=q|\chi_\phi(v)|^2f_v,
 \label{eq:field-labelled-spectrum}
\end{equation}
up to the harmless sign choice in Eq.~\eqref{eq:symplectic-form}.

The labelled identity above holds for every prime power.  For the full
perturbed-Alltop spectrum, now assume \(q=p^r\) with \(p\ge5\), and write
\begin{equation}
 \ket{\phi_{q,t}}=
 \frac{\ket{A_q}+t\ket0}{\sqrt{N_{q,t}}},
\end{equation}
and introduce the quadratic character \(\vartheta\) of \(\F_q^\times\) and
\begin{equation}
 \gamma_q=q^{-1/2}\sum_{x\in\F_q}\psi(x^2),
 \qquad |\gamma_q|=1.
 \label{eq:gamma}
\end{equation}
The standard quadratic Gauss-sum identity \cite[Chap.~5]{LidlNiederreiter1997}
\begin{equation}
 \sum_x\psi(Ax^2+Bx)
 =\gamma_q\sqrt q\,\vartheta(A)
 \psi\!\left(-\frac{B^2}{4A}\right),
 \qquad A\ne0,
 \label{eq:quadratic-Gauss}
\end{equation}
follows by completing the square and evaluating the remaining quadratic Gauss sum.  For \(a\ne0\), substitution into Eq.~\eqref{eq:Alltop-Gauss} and inclusion of the spike terms yields
\begin{align}
 q|\chi_{\phi_{q,t}}(a,b)|^2
 =\frac1{N_{q,t}^2}\Bigg|&\gamma_q\vartheta(-3a)
 \psi\!\left(\frac{b^2}{12a}-\frac{ab}{2}-\frac{a^3}{4}\right)
 \notag\\[-1mm]
 &+t\psi(-a^3)[1+\psi(-ab)]\Bigg|^2.
 \label{eq:complete-off-axis}
\end{align}
For \(a=0,b\ne0\),
\begin{equation}
 q|\chi_{\phi_{q,t}}(0,b)|^2
 =\frac{(\sqrt q\,t^2+2t)^2}{N_{q,t}^2},
 \label{eq:complete-axis}
\end{equation}
and the identity eigenvalue is \(q\).  Equations~\eqref{eq:complete-off-axis} and \eqref{eq:complete-axis} are a closed description of the full projector-Gram spectrum for every extension field of characteristic at least five.

\section{Numerical protocol}
\label{app:numerics}

The numerical data were regenerated with Python 3.14.2, NumPy 2.5.2, SciPy 1.18.0, and Matplotlib 3.11.1.  A deterministic workflow regenerates both benchmark data sets and the two main figures; the principal parameters are listed in Table~\ref{tab:numerics}.

\begin{table}[t]
\caption{Reproducibility parameters for Figs.~\ref{fig:story} and \ref{fig:operational-MSE}.}
\label{tab:numerics}
\begin{ruledtabular}
\begin{tabular}{lc}
Quantity & Value\\
\hline
Haar seed & 20260810\\
Fig.~\ref{fig:story} Haar samples per \(d\) & 2000\\
Fig.~\ref{fig:story} Haar dimensions & \(2\)--\(50\)\\
Fig.~\ref{fig:operational-MSE} Haar samples per \(q\) & 4000\\
Operational prime dimensions & \(5,7,11,13,17,19\)\\
Reported quantiles & 0.10, 0.50, 0.90\\
Alltop prime powers & \(q\le10^4\), \(p\ge5\)\\
BIC representative radius & \(\rho=1/3\)\\
BIC phase \(t_d\) & \(0\) (odd), \(1/(4d)\) (even)
\end{tabular}
\end{ruledtabular}
\end{table}

The plotted geometric representative uses \(\alpha=(1/3)e^{2\pi i t_d}\) with the phase schedule in Table~\ref{tab:numerics}; its spectrum is checked for strict positivity.  The displayed upper envelope is the supremum \(\rho\uparrow1/2\) in Eq.~\eqref{eq:geometric-bound} and bounds only that explicit sufficient region.  For the operational benchmark, a direct \(d=3\) frame-channel calculation verifies Eq.~\eqref{eq:linear-MSE}; all plotted structured fiducials are checked for the Moyal sum, strict invertibility, and the Alltop minimum.  Deterministic \(\F_{25}\) and \(\F_{49}\) tests verify the extension-field Moyal identity, attained spectral interval, and, for \(\F_{25}\), the complete Gauss phase.

\end{document}